\documentclass[12pt,a4paper]{article}
\usepackage[margin=2.2cm]{geometry}
\usepackage{authblk}

\usepackage[breaklinks]{hyperref}
\usepackage{booktabs}
\usepackage{newtxtext,newtxmath}

\usepackage{amsthm,mathbbol}
\numberwithin{equation}{section}
\newtheorem{theorem}{Theorem}
\newtheorem{proposition}[theorem]{Proposition}

\newtheorem{corollary}[theorem]{Corollary}

\newtheorem*{remark*}{Remark}

\numberwithin{theorem}{section}

\usepackage{tikz}
\usetikzlibrary{matrix,arrows,patterns,cd}
\usepackage{leftidx}
\usepackage{amsmath,amssymb,amsfonts,mathbbol}
\usepackage{slashed}
\newcommand{\bbLbrack}{[\kern-0.4em{[}\,}
\newcommand{\bbRbrack}{\,]\kern-0.4em{]}}

\usepackage{extarrows}

\usepackage{mathrsfs,latexsym}
\usepackage[mathscr]{eucal}

\newcommand{\II}{{\mathbb{1}}}

\newcommand{\CC}{{\mathbb C}}

\newcommand{\RR}{{\mathbb R}}

\newcommand{\NN}{{\mathbb N}}

\newcommand{\CoinX}[1]{C_0^\infty({#1})}

\newcommand{\fc}{\textit{fc}}
\newcommand{\pc}{\textit{pc}}

\newcommand{\Ac}{{\mathcal A}}
\newcommand{\Bc}{{\mathcal B}}
\newcommand{\Cc}{{\mathcal C}}

\newcommand{\Uc}{{\mathcal U}}
\newcommand{\Xc}{{\mathcal X}}

\newcommand{\Prob}{{\rm Prob}}

\DeclareMathOperator{\supp}{supp}

\renewcommand{\Re}{{\rm Re}\,}

\newcommand{\im}{{\rm im}\,}

\newcommand{\dvol}{d\textrm{vol}}

\newcommand{\Lb}{{\boldsymbol{L}}}
\newcommand{\Mb}{{\boldsymbol{M}}}
\newcommand{\Nb}{{\boldsymbol{N}}}

\newcommand{\Af}{{\mathscr A}}

\newcommand{\Bf}{{\mathscr B}}

\newcommand{\If}{{\mathscr I}}

\DeclareMathOperator{\Sol}{Sol}

\newcommand{\id}{{\rm id}}

\newcommand{\Var}{{\rm Var}}

\newcommand{\sym}{{\rm sym}}

\makeatletter
\newcommand{\raisemath}[1]{\mathpalette{\raisem@th{#1}}}
\newcommand{\raisem@th}[3]{\raisebox{#1}{$#2#3$}}
\makeatother

\begin{document}

\def\utilde#1{\mathord{\vtop{\ialign{##\crcr
   $\hfil\displaystyle{#1}\hfil$\crcr\noalign{\kern1.5pt\nointerlineskip}
   $\hfil\tilde{}\hfil$\crcr\noalign{\kern1.5pt}}}} {}}

\title{Quantum fields and local measurements} 

%
%
%
%

\author[1]{Christopher J. Fewster\thanks{\tt chris.fewster@york.ac.uk}}
\affil{Department of Mathematics,
	University of York, Heslington, York YO10 5DD, United Kingdom.}
\author[2]{Rainer Verch\thanks{\tt rainer.verch@uni-leipzig.de}}
\affil{Institute for Theoretical Physics, University of Leipzig, 04009 Leipzig, Germany.}

\maketitle 
%
%
\begin{abstract}   
The process of quantum measurement is considered in the algebraic framework of quantum field theory on curved spacetimes. Measurements are carried out on one quantum field theory, the ``system'', using another, the ``probe''. The measurement process involves a dynamical coupling of ``system'' and ``probe'' within a bounded spacetime region. The resulting ``coupled theory'' determines
a scattering map on the uncoupled combination of the ``system'' and ``probe'' by reference to natural ``in'' and ``out'' spacetime regions. No specific interaction is assumed and all constructions are local and covariant.

Given any initial state of the probe in the ``in'' region, the scattering map 
determines a completely positive map from ``probe'' observables in the ``out'' region to ``induced system observables'', thus providing a measurement scheme for the latter. It is shown that the induced system observables may be localized in the causal hull 
of the interaction coupling region and are typically less sharp than the probe observable, but more sharp than the actual measurement on the coupled theory.
Post-selected states conditioned on measurement outcomes are obtained using Davies--Lewis instruments that depend on the initial probe state. Composite measurements involving causally ordered coupling regions are also considered. Provided that the scattering map obeys a causal factorization property, the causally ordered composition of the individual instruments coincides with the composite instrument;
in particular, the instruments may be combined in either order if the coupling regions are causally disjoint. This is the central consistency property of the proposed framework. 

The general concepts and results are illustrated by an example in which both ``system'' and ``probe'' are quantized linear scalar fields, coupled by a quadratic interaction term with compact spacetime support. System observables induced by simple probe observables are calculated exactly, for sufficiently weak coupling, and compared with first order perturbation theory. 
\end{abstract} 

\bigskip
\noindent\emph{Dedicated to the memory of Paul Busch.}  

\section{Introduction}\label{sec:intro} 

This paper combines ideas and methods from algebraic quantum field theory (AQFT) and quantum measurement theory (QMT) in order to provide improved operational foundations for the measurement theory of relativistic quantum fields in (possibly curved) spacetimes. The aim is to provide a framework that is both conceptually clear and amenable to practical computations. In so doing, we bridge a gap between these subjects that has, surprisingly, lain open for a long time, despite its clear 
relevance to important discussions concerning the Unruh and Hawking effects~\cite{Unruh:1976,Hawking:1975}. 
On one hand, algebraic quantum field theory~\cite{Haag} is founded on the idea of algebras of observables associated with local regions of spacetime. However, not much attention has been given to how these observables can actually be measured. On the other hand, quantum measurement theory~\cite{Busch_etal:quantum_measurement} provides an operational understanding of measurement schemes, in which a probe system is used to measure a quantum observable of the system of interest. However these discussions are not usually framed in a spacetime context. By contrast, this paper will introduce a generally covariant formalism of measurement schemes adapted to algebraic quantum field theory in curved spacetimes, illustrated by a specific model that can be analysed in detail.  

The main work on measurement of local observables of which we are aware is due to Hellwig and Kraus~\cite{HellwigKraus:1969,HellwigKraus:1970,HellwigKraus_prd:1970}. 
In \cite{HellwigKraus_prd:1970},
one of their main points of focus was the question of where a state-reduction might be considered to occur in a relativistic model, 
given that the instantaneous reductions of quantum mechanics break manifest Lorentz covariance; they did not discuss probes as such 
but simply took as their starting-point the standard measurement-induced state reduction as described by L\"uders' rule~\cite{Busch2009} along with the locality 
and covariance of the quantum field theory (QFT).
In the refs.\ \cite{HellwigKraus:1969,HellwigKraus:1970}, Hellwig and Kraus considered a quantum (field) system dynamically coupled to an apparatus (or probe), assuming that the 
dynamics can be described in an interaction picture by a unitary $S$-matrix, 
with assumed locality properties, from which they inferred locality properties of the field observables. Some of 
these results can be seen as forerunners to ours, however, our framework is more general in several respects, as our approach is not restricted to Minkowski spacetime
and our assumptions on the dynamics are more general and do not require the existence of a unitary $S$-matrix describing the dynamics of the interaction between a 
quantum system and a probe (or apparatus); furthermore, the probes we discuss are physical systems in spacetime which allows 
us to address their localisation as well. We also benefit from various more recent developments in QMT. 
More recent work discussing the measurement process in quantum field theory includes~\cite{OkamuraOzawa,DoplicherQFM}, though neither
reference models the interaction between system and probe; the same is true of the 
insightful review of Peres and Terno~\cite{PeresTerno:2004}. 

Since the work of Hellwig and Kraus, there has been much progress in both QFT and quantum measurement theory. In particular, the operational approach to quantum measurement has been developed in considerable depth and
detail~\cite{DaviesLewis:1970,Davies_QTOS:1976,Ozawa:1984,Busch_etal:quantum_measurement}. Progress in QFT over the same period has brought many advances both in its 
phenomenology and its mathematical and conceptual underpinnings; in particular, the entire subject of QFT in curved spacetimes (QFT in CST) has developed to a mature state.  
Further understanding and development of QFT in CST brings with it the need to accommodate the description of measurement process in a covariant spacetime context.
For instance, there are varying interpretations of the famous Unruh effect~\cite{Unruh:1976,DeWitt:1979}, which also has a relation to Hawking radiation from black 
holes~\cite{Hawking:1975}, bearing on the question of what a particle detector is and what a particle might be.
The traditional approach to these questions centres on the Unruh-deWitt detector, in which a quantum mechanical probe system is coupled to the quantum field. 
If the detector executes uniform accelerated motion, then the probe system is known to become excited; see~\cite{Unruh:1976} and~\cite{deBievreMerkli:2006} for a deep, 
general and rigorous account. However the behaviour of the probe is not, to our knowledge, ever analysed in terms of statements concerning local observables of the quantum field itself.\footnote{As the final version of this paper was prepared, we became aware of a thesis by Smith~\cite{Smith_thesis} (written contemporaneously with our work) which considers the specific question of what observables of a field are measured by an Unruh-deWitt detector, although without developing the general framework that we will provide here. We thank ARH Smith for bringing his work to our attention.}
(In a complementary approach \cite{FredenhagenHaag:1987,FredenhagenHaag:1990}, certain observables of the quantum field are used to describe the Unruh effect and the 
Hawking effect and they are referred to as ``detector observables'', but without any coupling of the quantum field to a probe system as in the references cited before.)
 In the light of recent discussions concerning the thermal nature of the
Unruh effect~\cite{BuchholzSolveen:2013,BucVer_macroscopic:2015,BuchholzVerch:2016}, it seems desirable to establish a clear and systematic account of how probes may be used to
measure local properties of a quantum field together with their relation to observables of the quantum field.
This is what we will do, intending that our discussion will be accessible to workers in both QFT and quantum measurement theory. 

To be clear on what we do not do:  we do not attempt to discuss measurement in quantum gravity, but consider a fixed, possibly curved spacetime in the sense of macroscopic physics. 
We also do not claim to solve the measurement problem of quantum theory. Rather, we take it for granted that the experimenter has some means of preparing, controlling and measuring 
the probe and sufficiently separating it from the QFT of interest -- which we will call the `system' -- the question is what measurements of the probe tell us about the system. That is, our interest is in describing a link in the measurement chain, in a covariant spacetime context.
We also do not attempt to prove that all local observables of a QFT can, in fact, be measured using a suitable probe 
(see \cite{Ozawa:1984,OkamuraOzawa} and literature cited there for results in this direction for quantum mechanics and QFT in flat spacetime). 

We can now describe what we will do. After some brief preliminaries, 
we set out, in section~\ref{sec:genscheme}, a general framework in which two physical systems, the `system' and the `probe', may be coupled 
together in a fashion suitable for measurement. In particular, the probe and system are prepared in known states $\sigma$ and $\omega$ at 
`early times', during which they are uncoupled; they are again uncoupled at `late times', during which an observable $B$ of the probe is measured. 
Here, the coupling is taken to be effective only in a compact spacetime `coupling region', while `early times' and `late times' refer to covariantly defined spacetime regions. As we show, the expected value of the resulting measurement coincides with the expected value of an \emph{induced system observable} $\varepsilon_\sigma(B)$ in a hypothetical measurement in the system state $\omega$. However, the variance of the actual measurement typically exceeds that of the hypothetical measurement, due to detector fluctuations.
Under reasonable assumptions concerning the coupling we show that $\varepsilon_\sigma(B)$ may be localised in suitable neighbourhoods of 
the causal hull of the coupling region, regardless of what the localisation of $B$ is and where the measurement reading is taken. However, if $B$ may be localised in the causal complement of $K$, the induced system observable is a multiple of the unit, from which no information concerning the system may be extracted.
We also give an account of \emph{effect valued measures} (EVMs) in this framework and explain how joint measurements of probe EVMs at spacelike separation can provide natural examples of joint unsharp measurements of non-commuting system observables.

Next, we discuss selective and non-selective measurements, introducing the concept of a pre-instrument as the map that sends system states to post-selected states conditioned on the observation of an effect. (The term `post-selected' is used in various different ways in the literature -- the precise meaning we have in mind, which amounts to updating the state based on the measurement outcome, will be spelled out in detail.) In particular, we show that, at spacelike separation from the coupling region, the original and post-selected states agree only on observables that are uncorrelated, in the original state, with the system observable induced by the measured probe effect.
As QFT states typically exhibit correlations even at spacelike separation, it
is clear that every spacetime region typically contains observables whose expectation values differ in the two states. By analysing successive measurements with couplings in causally ordered regions we show that the post-selection may be performed sequentially in any valid causal order, or in a combined single stage, with the same outcome. In particular, where the coupling regions are causally disjoint the post-selection may be performed in either order.
Elsewhere~\cite{BostelmannFewsterRuep:2020}, this analysis is used to cast light on the `impossible measurements' raised many years ago by Sorkin~\cite{sorkin1993impossible}.

We also briefly discuss the significance of geometric and internal symmetries in our framework, although more could certainly be said on both those subjects. Here we show that if there is a global gauge group acting on both system and probe, under which the coupling transforms covariantly, then gauge invariant probe observables induce gauge invariant system observables. Turning this around, gauge noninvariant system quantities can only be measured using gauge-breaking couplings or probes. On the geometrical side, we point out that when a system state has a \emph{strong mixing}
cluster property under a time-translation symmetry, then the post-selected state becomes eventually indistinguishable from the original after an elapse of time.

An important aspect of our treatment is that it is amenable to concrete calculations.
Sections~\ref{sec:probe} and~\ref{sec:applic} are devoted to a specific system--probe model consisting of two free real scalar fields with a quadratic interaction between them with a spacetime dependent coupling factor of compact support. In order to prepare the ground for the general analysis of section~\ref{sec:genscheme} we now preview these results in some detail. Some fine points of precision
are suppressed in this description and we emphasise that our framework is not tied to this example but applies to general QFTs including those with self-interactions.

Consider two linear quantum fields $\Phi$ and $\Psi$ described by the uncoupled classical action 
\begin{equation}
S_0 = \frac{1}{2}\int_M \dvol \left((\nabla_a \Phi)(\nabla^a\Phi) - m_\Phi^2 \Phi^2 +  (\nabla_b \Psi)(\nabla^b\Psi) - m_\Psi^2 \Psi^2\right),
\end{equation}
where $m_\Phi$ and $m_\Psi$ are the masses of the two fields and $\Mb$ is a globally hyperbolic spacetime. In what follows, $\Phi$ will be the `system' field and $\Psi$ will be the `probe'. A coupling between them can be introduced by adding an interaction term
\begin{equation}
S_{\text{int}}= - \int_\Mb\dvol\, \rho \Phi  \Psi
\end{equation}
to the action, where $\rho$ is a real, smooth function with support contained in a compact set $K$. The Euler--Lagrange field equations for the uncoupled and coupled systems can be written respectively as
\begin{equation} 
\begin{pmatrix}
	P & 0 \\ 0 & Q
\end{pmatrix}\begin{pmatrix}
	\Phi\\ \Psi
\end{pmatrix}
=0,\qquad \begin{pmatrix}
P & R \\ R & Q
\end{pmatrix}\begin{pmatrix}
\Phi\\ \Psi
\end{pmatrix}
=0,
\end{equation}
where 
\begin{equation}
P=\Box_\Mb+m_\Phi^2, \qquad Q=\Box_\Mb+ m_\Psi^2, 
\end{equation}
and $R$ is the operation of multiplication by $\rho$.  

The two theories may be quantized by standard methods, introducing algebras  $\Uc(\Mb)$ generated by smeared fields $\Phi(f)$, $\Psi(h)$ for 
the uncoupled theory and an algebra $\Cc(\Mb)$ generated by $\Phi_{\text{int}}(f)$ and $\Psi_{\text{int}}(h)$ for the coupled one. Here $f$ and $h$ are
smooth compactly supported test functions. The generators obey various relations that will be spelled out in full later on. This way of modelling `system' and probe', and in particular, their coupling, is in the spirit of the `standard model' of a quantum measurement process as discussed in \cite{BuschLahti-StandModel}; it appears also in 
discussions of the Unruh effect taking as `probe' a quantum field on Minkowski spacetime \cite{Unruh:1976} or in a cavity \cite{GroveOttewill:1983}.

The coupled and uncoupled theories may be identified in the natural `in' and `out' regions $M^-$ and $M^+$, defined by $M^\pm=M\setminus J^\mp(K)$, where $J^{+/-}(K)$ denotes the causal future/past of $K$ (see section~\ref{sec:prelim}). Formally, this identification is
implemented by algebraic isomorphisms $\tau^\pm:\Uc(\Mb)\to\Cc(\Mb)$ so that
\begin{equation}
\tau^\pm\Phi(f) =\Phi_{\text{int}}(f) , \qquad
\tau^\pm\Psi(f) =\Psi_{\text{int}}(f) 
\end{equation}
for all test functions $f$ supported in $M^\pm$. These identifications can be compared by a scattering map 
\begin{equation}
\Theta=(\tau^{-})^{-1}\tau^+,
\end{equation}
which is an isomorphism of the uncoupled algebra to itself.

We consider a measurement of the coupled probe-system theory in a state $\varpi$ of $\Cc(\Mb)$ which has no correlations between the two theories at `early times', meaning that $(\tau^-)^*\varpi=\omega\otimes\sigma$. Here, $\omega$ and $\sigma$ are the states in which the system and probe have been individually prepared at early times. The measured observable is the smeared field $\Psi_{\text{int}}(h)$, where $h$ is supported in $M^+$. As $\Psi_{\text{int}}(h)=\tau^+\Psi(h)$, this measurement may be considered as an observation of the probe at `late times'.
	 
The expectation value of the measurement outcome is $\varpi(\Psi_{\text{int}}(h))$. Although the measurement is performed on the coupled system, one wishes to interpret the result as a measurement on the system itself. This is possible if there is a system observable $A$, depending perhaps on $\sigma$ but not on $\omega$, for which
\begin{equation}
\omega(A) = \varpi(\Psi_{\text{int}}(h))
\end{equation}
and so that $A$ is the unique observable with this property for all $\omega$. In this case $A$ will be called an \emph{induced system observable}. One of the goals of this paper is to show how induced system observables may be introduced in a general setting, to determine their localisation properties, and to compute them in the model described above. This computation, which makes use of the scattering map $\Theta$, shows that the system observable induced by $\Psi_{\text{int}}(h)$ is
\begin{equation}
A = \Phi(f^-) + \sigma(\Psi(h^-))\II
\end{equation}
for test functions $f^-$ and $h^-$, depending on $h$, so that $f^-$ and the difference $h^--h$ are supported in the intersection of coupling region $\supp\rho$ with the causal past of the support of $h$. The dependence of $A$ on $\sigma$ is to be expected; note also that if $h$ lies completely outside the causal future of $\rho$ then $A=\sigma(\Psi(h))\II$ is
a trivial observable, from which one can learn nothing about the system.

The compactness of $\supp f^-$ indicates that the observable $A$ is local. However, as we will argue, the observable $A$ has properties that are not
local to the support of $f^-$ (unless this set happens to be causally convex). 
Instead, we will show that $A$ can be appropriately localised within regions that
contain the \emph{causal hull} (sometimes called the causal completion) of $\supp f^-$; that is, the intersection of its causal future and past. For dynamical reasons any 
local observable is localisable in many regions -- in particular, in neighbourhoods of any Cauchy surface -- but the localisation close to the coupling region provides a particularly attractive physical picture of the measurement. Note that $\supp h$ may be located far from the localisation region of the induced observable.

Further analysis of this model appears in Secs.~\eqref{sec:probe} and~\eqref{sec:applic}. Among other things, we show that the scattering morphism satisfies the causal factorization property where multiple couplings are concerned, that the set of induced system observables forms a subalgebra of the algebra of smeared system fields, and that the results replicate those of first order perturbation theory in an appropriate limit. Section~\ref{sect:conc} gives some final remarks and the four appendices address technical points arising in the text.

\section{Preliminaries}
\label{sec:prelim}

\paragraph{Background on Lorentzian geometry}
A Lorentzian spacetime will be a smooth (Hausdorff, paracompact) manifold $M$ with at most finitely many connected components, equipped with a smooth
Lorentzian metric $g$ of signature $+-\cdots -$ and a choice of time-orientation, thus allowing all nonzero causal [timelike or null] vectors to be classified as 
future- or past- pointing. If $x\in M$, the \emph{causal future/past} $J^{+/-}(x)$ of $x$ is the set of all points reached from $x$ by smooth
future-directed causal curves (including $x$ itself); for a subset $S\subset M$, we write $J^\pm(S)=\bigcup_{x\in S}J^\pm(x)$ and also $J(S)=J^+(S)\cup J^-(S)$. 
The \emph{causal hull} of $S\subset M$ is the intersection $J^+(S)\cap J^-(S)$; that is, the set of all points that lie on causal curves with both endpoints in $S$.
A subset is \emph{causally convex} if it is equal to its causal hull, and therefore contains every causal curve that begins and ends in it. One may easily show that the causal hull
of $S$ is the intersection of all causally convex sets containing $S$. 

The spacetime is \emph{globally hyperbolic} if and only if it is devoid of closed causal curves and the causal hull of any compact set is compact~\cite{Bernal:2006xf,Minguzzi:2013}.
A \emph{Cauchy surface} is a set intersected exactly once by every inextendible smooth timelike curve; every Lorentzian spacetime possessing a Cauchy surface is globally hyperbolic, and every globally hyperbolic spacetime may be foliated into Cauchy surfaces that are, additionally, smooth spacelike hypersurfaces. We usually denote a globally hyperbolic spacetime by a single symbol $\Mb$, incorporating the underlying manifold, metric and time orientation. Any open causally convex subset of a globally hyperbolic spacetime is itself globally hyperbolic, when equipped with the induced metric and time-orientation.

The \emph{causal complement} of a set $S$ is defined as $S^\perp=M\setminus J(S)$, 
and sets $S$ and $T$ are \emph{causally disjoint} if $T\subset S^\perp$ or equivalently $S\subset T^\perp$, i.e., if there is no causal curve joining $S$ and $T$. 
In a globally hyperbolic spacetime, the causal future and past of an open set are open, while those of a compact set are closed; accordingly, if $K$ is compact then $K^\perp$ is open and $K^{\perp\perp}$ is closed [though not necessarily compact] and contains $K$. Note that $K^{\perp\perp}$ is not, in general, the causal hull of $K$ although there are situations in which they do coincide.\footnote{In Minkowski space, an example where they coincide is when $K$ is a timelike curve segment; an example where they differ is when $K$ is a subset of a constant time hypersurface.}  If $J^+(S)\cap J^-(T)$ is empty (or, equivalently, if $J^+(S)\cap T$ or $S\cap J^-(T)$ are empty), for subsets $S$ and $T$, then there is a Cauchy surface of $\Mb$ lying to the future of $T$ and the past of $S$, thus establishing a causal ordering in which $S$ is later than $T$. In the case where $S$ and $T$ are causally disjoint, it is possible to order $S$ both later and earlier than $T$. For this reason, 
if one or both of $J^+(S)\cap J^-(T)$ or $J^-(S)\cap J^+(T)$ are empty, we say that $S$ and $T$ are \emph{causally orderable}. Finally, the future/past \emph{Cauchy development} $D^{+/-}(S)$ of a set $S\subset M$ is the set of points $p$ so that every past/future-inextendible piecewise smooth causal curve through $p$ meets $S$, and $D(S)=D^+(S)\cup D^-(S)$.
See e.g.,~\cite[Appx~A]{FewVer:dynloc_theory} for some relevant proofs, references and further discussion.

\paragraph{Background on algebraic QFT}

We summarise some basic ideas of algebraic QFT, which is the framework we adopt for
our discussion. See Haag's classic exposition~\cite{Haag} and the recent book~\cite{AdvAQFT} for details, 
and~\cite{FewsterRejzner_AQFT:2019} for a pedagogical introduction. Our viewpoint is particularly influenced by locally covariant QFT~\cite{BrFrVe03,FewVerch_aqftincst:2015}
but we will avoid having to introduce all the structures of this approach.

Fix a particular QFT, which we will label $\Ac$, that is defined on some collection of globally hyperbolic spacetimes. To each spacetime $\Mb$ in this collection, the theory 
should specify a unital $*$-algebra $\Ac(\Mb)$ and a collection of sub-$*$-algebras $\Ac(\Mb;N)$ labelled by the causally convex open subsets $N$ of $\Mb$, with all these subalgebras containing the unit of $\Ac(\Mb)$. Some remarks on the operational interpretation of these algebras appear below.

We make five assumptions, which we will assume to hold of any AQFT $\Ac$ unless explicitly stated otherwise. The first is called \emph{isotony}: if 
$N_1\subset N_2$ then $\Ac(\Mb;N_1)\subset \Ac(\Mb;N_2)$. The second, \emph{compatibility}, requires that if $N$ is an open causally convex subset of a spacetime $\Mb$ on which $\Ac$ is defined, then $\Ac$ is also defined on $\Nb$ and there is an injective unit-preserving algebraic $*$-homomorphism $\alpha_{\Mb;\Nb}:\Ac(\Nb)\to\Ac(\Mb)$, whose image coincides with the subalgebra $\Ac(\Mb;N)$.\footnote{Where spacetimes with nontrivial topology are concerned, the injectivity assumption may have to be relaxed for some theories. See~e.g.,~\cite{SandDappHack:2012,BeniniDappiaggiSchenkel:2013}.}  Here, $\Nb$ is the globally hyperbolic spacetime comprising $N$ with the metric and time-orientation inherited from $\Mb$.
It is further required that these maps,  which we will refer to as \emph{morphisms} for brevity, obey
\begin{equation}\label{eq:functorial}
\alpha_{\Mb_1;\Mb_2}\circ \alpha_{\Mb_2;\Mb_3} = \alpha_{\Mb_1;\Mb_3}
\end{equation}
if $M_3\subset M_2\subset M_1$.

Third, the \emph{time-slice property} requires that one has $\Ac(\Mb;N)=\Ac(\Mb)$ (equivalently, that $\alpha_{\Mb;\Nb}$ is an isomorphism) whenever $N$ contains a Cauchy surface for $\Mb$. Combining this with compatibility, we see that $\Ac(\Mb;N_1)=\Ac(\Mb;N_2)$ if $N_1\subset N_2$ and $N_1$ contains a Cauchy surface for $\Nb_2$.

Fourth, we assume that \emph{Einstein causality} holds: if regions $N_1$ and $N_2$ within $\Mb$ are causally disjoint then the elements of $\Ac(\Mb;N_1)$ commute with the elements of $\Ac(\Mb;N_2)$.

Finally, we add an assumption that we call the \emph{Haag property}.
Let $K$ be a compact subset of $\Mb$. Suppose that an element $A\in\Ac(\Mb)$ commutes with every element of $\Ac(\Mb;N)$ for every region $N$ contained in the causal complement $K^\perp$ of $K$. Then we assume that $A\in\Ac(\Mb;L)$ whenever $L$ is a \emph{connected}\footnote{Demanding that this holds also when $L$ has multiple connected components may conflict with additivity and so we restrict to the connected case.} open causally convex subset containing $K$. This is a weakened form of Haag duality~\cite{Haag}.\footnote{Exact Haag duality is better phrased in the context of local von Neumann algebras; it also brings complications concerning the regularity of the regions to which it applies. See e.g.,~\cite{Camassa:2007}.}

\paragraph{Observables, states and operations}
The now-standard physical interpretation of AQFT  (though not the original interpretation -- see below) is that the self-adjoint elements $A=A^*$ of $\Ac(\Mb)$ are local observables of the theory, with the self-adjoint
elements of $\Ac(\Mb;N)$ being those observables that can be localised in $N$.
Due to the time-slice property and isotony, a given observable can have many different
localisation regions. One of our purposes in this paper is to provide this
interpretation with a better operational basis.
 
Actually, the viewpoint just sketched is slightly narrow. 
The usage of `observable' to mean a self-adjoint algebra element is parallel
to standard usage in quantum mechanics, where observables are usually identified with self-adjoint operators. In turn, each self-adjoint operator $A$ corresponds uniquely
to a projection valued measure $P_A$ defined on the Borel sets of $\RR$ (and supported on the spectrum of $A$) such that 
\begin{equation}\label{eq:specdecomp}
A = \int_{\RR} \lambda\,dP_A(\lambda)
\end{equation} 
and indeed 
\begin{equation}\label{eq:funcalc}
f(A) = \int_{\RR} f(\lambda)\,dP_A(\lambda) 
\end{equation}
for suitable functions $f$.
Conversely, any projection valued measure $P_A$ defines a self-adjoint operator $A$ by~\eqref{eq:specdecomp}. We recall that, for a projection valued measure $P$ defined on the Borel sets of $\mathbb{R}$, or on measurable subsets $X$ of a more general set $\Omega$, each $P(X)$ is a projection on a (fixed) Hilbert space, such that (i)  $\sum_jP(X_j) = P(X)$ for any disjoint decomposition
of $X$ into countably many $X_j$, (ii) $P(\Omega) = {\bf 1}$ and (iii) $P(X)P(X') = 0$ whenever $X$ and $X'$ have void intersection. A natural generalization
is a positive operator valued measure (or `effect valued measure', as we will later term it),  where $P(X)$ is a positive (more precisely, non-negative) operator for all $X$ satisfying the 
conditions (i) and (ii), but not (iii). While one can still associate a self-adjoint operator with the positive operator valued measure by~\eqref{eq:specdecomp}, there is no longer a functional calculus relation of the form~\eqref{eq:funcalc}. 

One of the lessons of quantum measurement theory is that quantum observables
should be viewed as corresponding to positive operator valued measures (see \cite{Busch_etal:quantum_measurement} for full discussion), with projection valued measures providing the special case of `sharp' measurements, as opposed to the 
general `unsharp' situation.
A conceptually important example of the use of positive operator valued measures is to describe time of arrival in quantum mechanics~\cite{Giannitrapani1997,BrunettiFredenhagen-time:2002,Werner-TOA}.

With these considerations in mind, $\Ac(\Mb;N)$ should be regarded as including all (evaluations of) positive operator valued measures of observables localisable in $N$ (only finite additivity is required in the $*$-algebraic setting). Nonetheless, the term `observable' in AQFT is so strongly associated with self-adjoint elements that it seems wise to adhere to this convention and refer
explicitly to positive operator or effect valued measures where appropriate.

In AQFT, states of the theory on $\Mb$ are linear functionals from $\Ac(\Mb)$ to $\CC$ that are positive, i.e., $\omega(A^*A)\ge 0$ for all $A\in\Ac(\Mb)$, and normalized so that $\omega(\II)=1$. The value $\omega(A)$ is the expected value of measurement outcomes when $A$ is measured in state $\omega$. We mention that, in the $*$-algebra context, an element is described as positive if it is a finite convex combination of elements of the form $A^*A$. Given a state, one may proceed to construct a Hilbert space representation using the GNS construction~\cite{Haag}; however, we will not need to do so at any point in our discussion.

In fact, Haag and Kastler~\cite{HaagKastler1964} were reluctant to interpret elements of the local algebras as observables (which they considered to arise as limits of local algebra elements). Instead, they viewed
the elements in $\Ac(\Mb;N)$ in terms of operations that could be performed within $N$ on the states of $\Ac(\Mb)$; specifically, each $B \in \Ac(\Mb;N)$ 
corresponds to the operation mapping $\omega(\,\cdot\,) \mapsto \omega(B^* \,\cdot\,B)/\omega(B^*B)$. A connection
between operations and positive operator valued measures, leading to the concept of `instruments', has been established in \cite{DaviesLewis:1970}; see also \cite{OkamuraOzawa} for further discussion
in the context of quantum field theory.

We emphasise that, while the interpretations mentioned are certainly consistent with the general conditions laid down on an AQFT, they do not purport to set out how exactly one measures an observable or performs an operation within a region of spacetime. Part of the purpose of this paper is to provide just such an account.

\section{General description of the measurement scheme}\label{sec:genscheme}
 
In a controlled experiment, the experimenter makes a change in the world and compares subsequent observations with what, on the basis of other observations, would have happened otherwise. So as to be able to discuss a variety of experiments, it is convenient to express the discussion in terms of the counterfactual world in which the interaction does not occur, rather than the actual world of the experiment, in which it does. This comparison of different dynamical evolutions lies at the heart of scattering theory and has appeared in locally covariant QFT in the guise of `relative Cauchy evolution'~\cite{BrFrVe03} which has strongly influenced
the general definition of a
coupled system that we describe in Section~\ref{sec:abstract}.
Measurement processes have long been described in terms of scattering~~\cite{HellwigKraus:1969,HellwigKraus:1970}; however, we
believe that our treatment has a stronger operational basis and is fully adapted to curved spacetimes. After that, we describe the measurement scheme --
the way in which probe observables may be considered as measuring local system observables -- 
and in particular discuss the localisation properties of the local observables in Sec.~\ref{sec:scheme}.
A central part of the work concerns the state change consequent upon measurements, set out in Sec.~\ref{sec:instruments},
where localisation is again to the fore, as is the consistency of the framework in relation to composite measurements. 
Finally, Sec.~\ref{sec:symm} contains a brief discussion of the role of internal and geometric symmetries in the measurement chain.

\subsection{Abstract formulation of the coupling between system and probe}\label{sec:abstract}

Consider two algebraic QFTs $\Ac$ and $\Bc$, using $\alpha_{\Mb;\Nb}$ and $\beta_{\Mb;\Nb}$ to denote the inclusion maps arising from spacetime subregions. We will think of $\Ac$ as the system and $\Bc$ as a probe. The combined theory, comprising independent copies of $\Ac$ and $\Bc$ without any cross-interaction, may be denoted $\Uc=\Ac\otimes\Bc$ and assigns to $\Mb$ the algebra $\Uc(\Mb)=\Ac(\Mb)\otimes\Bc(\Mb)$ and local subalgebras $\Uc(\Mb;N)=\Ac(\Mb;N)\otimes\Bc(\Mb;N)$, with inclusion maps $\alpha_{\Mb;\Nb}\otimes\beta_{\Mb;\Nb}$. This is the control situation. 
As we do not want to become too immersed in technical detail, we here assume that the algebras have discrete topology and use the algebraic tensor product. If they were $C^*$-algebras, there would
be a choice of tensor products, among which the minimal tensor product would have certain advantages~\cite{BrFrImRe:2014}.

An experiment may be described by a further theory $\Cc$ in which the system and probe are coupled together. The morphisms for local embeddings will be denoted $\gamma_{\Mb;\Nb}$. For simplicity we will assume that the coupling is operative only within a compact spacetime region $K$, meaning that the theory $\Cc$ should reduce to $\Ac\otimes\Bc$ outside the causal hull $J^+(K)\cap J^-(K)$ of $K$. Precisely, this means that for each open causally convex subset $L$ of $M\setminus(J^+(K)\cap J^-(K))$, there is an isomorphism
\begin{equation}
\chi_\Lb:\Ac(\Lb)\otimes\Bc(\Lb)\to \Cc(\Lb)
\end{equation}
and which is compatible with the locality structures of the two theories as follows:
whenever $L$ and $L'$ are both open causally convex subsets of $M\setminus(J^+(K)\cap J^-(K))$ with $L'\subset L$ then 
\begin{equation}\label{eq:naturality}
\begin{tikzcd}[column sep=large]
\Ac(\Lb')\otimes\Bc(\Lb') \arrow[r, "\alpha_{\Lb;\Lb'}\otimes\beta_{\Lb;\Lb'}"]\arrow[d,"\chi_{\Lb'}"] &
\Ac(\Lb)\otimes\Bc(\Lb) \arrow[d,"\chi_\Lb"]  \\
\Cc(\Lb') \arrow[r, "\gamma_{\Lb;\Lb'}"] & \Cc(\Lb)
\end{tikzcd}
\end{equation}
commutes. This expresses the equivalence of the theories not only at the level of the local algebras themselves, but also in terms of the relations between these algebras; it is closely related to the idea of equivalence between theories explored in local covariant QFT~\cite{BrFrVe03,FewVer:dynloc_theory,FewVerch_aqftincst:2015}. Note that we do not specify what the coupling is; merely that the theory is not assumed to be equivalent to the uncoupled theories in regions overlapping $J^+(K)\cap J^-(K)$. Thus we have a completely general description of the coupling process. Later, we will analyse a specific model with specific couplings, showing explicitly that the morphisms $\chi_\Lb$ exist with the properties above. Demonstrating their existence for more general interactions raises many of the usual problems of constructive quantum field theory, although the restriction of the interaction to a compact region considerably simplifies matters and at least perturbatively (cf.\ the perturbative AQFT programme~\cite{Rejzner_book}) it seems clear that our
description of the coupling is viable and general.

\begin{figure}
\begin{center}
 \includegraphics[scale=0.25]{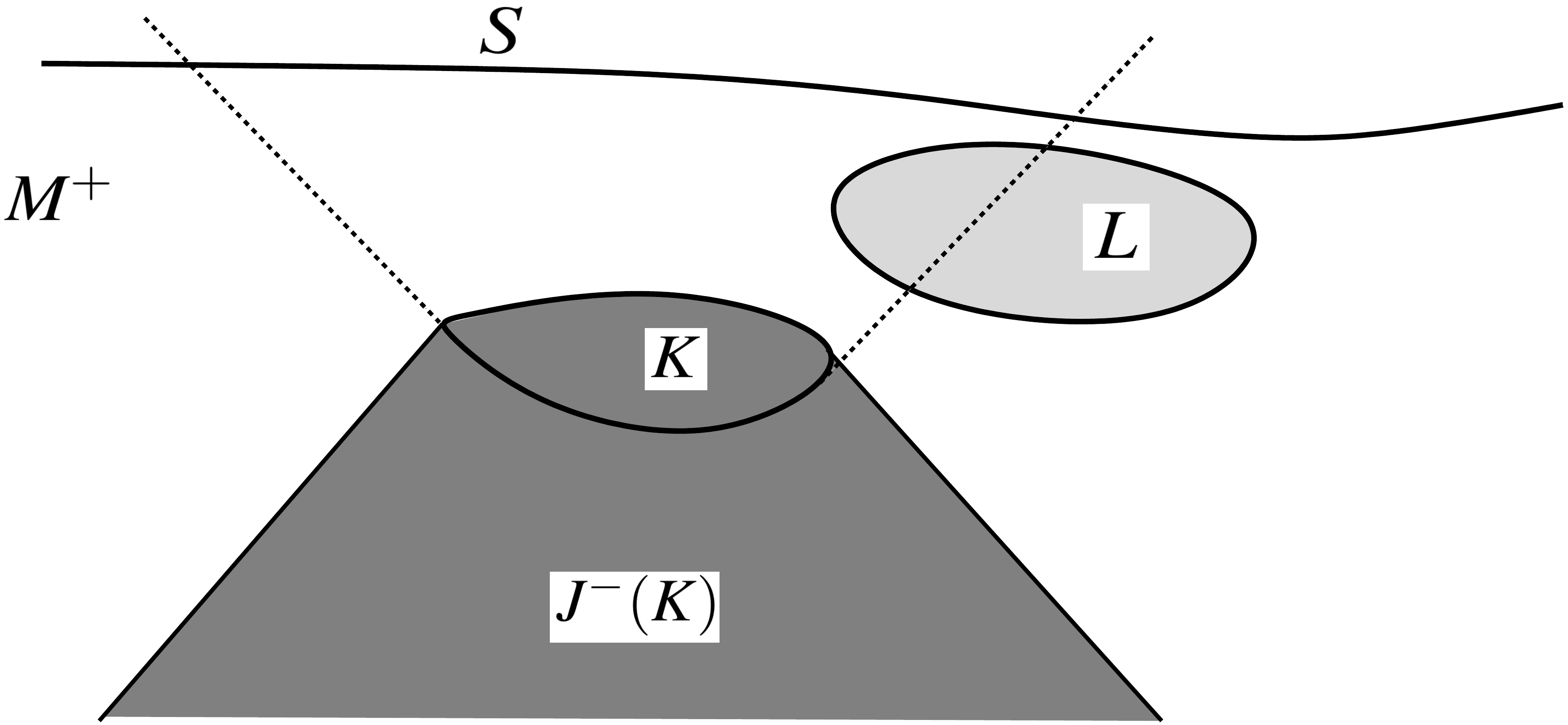} 
\end{center}
\caption{\small In this spacetime diagram, light rays are inclined at $45^\circ$ relative to the vertical axis, with the arrow of time pointing up the page. 
The spacetime region $K$ wherein the coupling takes place, together with its causal past $J^-(K)$ (which includes $K$), are shaded dark. The spacetime region $M^+$ is the complement of $J^-(K)$
and contains Cauchy-surfaces like $S$; $M^-$ is defined analogously. The lightlike boundary of $J^+(K)$ is indicated by dotted lines. The spacetime region $L$ lies outside 
the causal hull $J^+(K) \cap J^-(K)$ of $K$ (even though it intersects $J^+(K)$).
}\label{fig:spacetime}
\end{figure}
The coupling region $K$ determines natural `in' ($-$) and `out' ($+$) regions defined by $M^\pm=M\setminus J^\mp(K)$ which are open causally convex regions that together cover the exterior of $J^+(K)\cap J^-(K)$. 
See Figure~\ref{fig:spacetime} for an illustration. Note that these regions are determined covariantly once the interaction region is specified and without reference to any observer's clock or time coordinate.  We will use the morphisms associated with these regions quite frequently, and therefore abbreviate $\alpha_{\Mb;\Mb^\pm}$ and similar to $\alpha^\pm$, and $\chi_{\Mb^\pm}$ to $\chi^\pm$. As the regions $M^\pm$ contain Cauchy surfaces of $\Mb$ (see e.g.,~\cite[Lem.~A.4]{FewVer:dynloc_theory}), all of the morphisms $\alpha^\pm$, $\beta^\pm$, $\gamma^\pm$, $\chi^\pm$ are isomorphisms, as are the compositions 
\begin{equation}
\kappa^\pm = \gamma^\pm\circ \chi^\pm :\Ac(\Mb^\pm)\otimes\Bc(\Mb^\pm)\to \Cc(\Mb).
\end{equation} 
Using these maps, we define the retarded ($+$) and advanced ($-$) response maps 
\begin{equation}
\tau^\pm = \kappa^\pm\circ(\alpha^\pm\otimes\beta^\pm)^{-1},
\end{equation}
which identify the uncoupled system with the coupled one at early ($-$) or late ($+$) times. Combining these maps, one obtains a \emph{scattering morphism}
\begin{equation}\label{eq:Theta}
\Theta= (\tau^{-}){}^{-1}\circ \tau^+,
\end{equation} 
which is an automorphism of $\Ac(\Mb)\otimes\Bc(\Mb)$ and
would correspond to the adjoint action of the $S$-matrix in standard formulations of scattering theory. Note that $\Theta$ maps algebra elements from late times to early times. The following important properties of the scattering morphism are proved in Appendix~\ref{appx:genscattering}.
\begin{proposition}\label{prop:scattering} 
	(a) If $\hat{K}$ is any compact set containing the coupling region $K$, let $\hat{\Theta}$ be the morphism obtained if one replaces $K$ by $\hat{K}$ in the construction of the scattering morphism of a given coupled theory. Then $\hat{\Theta}=\Theta$.
	
	(b) If $L$ is an open causally convex subset of the causal complement $K^\perp = M^+\cap M^-$ of $K$, then $\Theta$ acts trivially on $\Uc(\Mb;L)=\Ac(\Mb;L)\otimes \Bc(\Mb;L)$.  
	
	(c)	Suppose that $L^+$ (resp. $L^-$) is an open causally convex subset of $M^+$ (resp., $M^-$), and that $L^+\subset D(L^-)$. Then
	$\Theta \Uc(\Mb;L^+) \subset \Uc(\Mb;L^-)$. 
\end{proposition}
Part (a) shows that the scattering morphism is canonically associated with the theories $\Uc$ and $\Cc$ and the identifications between them, while parts (b) and (c) are locality properties. In particular, part~(b) shows that, as one would expect, the coupling has no effect on observables localised in $L\subset K^\perp$. The more refined information provided by part~(c) plays an important role in~\cite{BostelmannFewsterRuep:2020}.
 
\subsection{The measurement scheme} \label{sec:scheme}
 
The isomorphisms just discussed allow us to express certain observables and states of the coupled theory in `uncoupled language'. Thus, a state $\varpi$ of $\Cc(\Mb)$ may be described as uncorrelated at `early times' (i.e., in $\Mb^-$) if $(\kappa^{-})^* \varpi$ is a product state over $\Ac(\Mb^-)\otimes\Bc(\Mb^-)$, or 
equivalently, if $(\tau^-)^*\varpi$ is a product state over $\Ac(\Mb)\otimes\Bc(\Mb)$.\footnote{Here the star denotes the adjoint map: in general $(\zeta^*\varpi)(X) = \varpi(\zeta (X))$.}  Similarly,  
probe observables measured at `late times' (i.e., in $\Mb^+$) are precisely
those of the form $\kappa^+(\II\otimes B)$ for $B\in\Bc(\Mb^+)$ or, equivalently, of the form $\tau^+ (\II\otimes B)$ for $B\in\Bc(\Mb)$.
Of course one can identify states that are uncorrelated at late times and probe observables that are measured at early times in analogous ways.
However, these are of less interest to us because the measurement process requires one to prepare early and measure late, relative to the interaction.

\paragraph{Induced system observables}
Suppose, now, that the probe is prepared in a state $\sigma$ of $\Bc(\Mb)$, while the field system is in state $\omega$ of $\Ac(\Mb)$. This situation corresponds to the combined state 
\begin{equation}
\utilde{\omega}_\sigma =  (\tau^-){}^{-1*}(\omega\otimes\sigma)
\end{equation}
on $\Cc(\Mb)$, which is uncorrelated at early times according to our discussion above. 
Let $B\in\Bc(\Mb)$ be an observable of the probe system, which can be identified at late times with  
the observable 
\begin{equation}\label{eq:Btilde}
\widetilde{B}=  \tau^+(\II\otimes B)
\end{equation}
of $\Cc(\Mb)$.  The expectation of the probe observable, when the system and probe have been prepared as described, may be written
\begin{equation}\label{eq:probe_exp}
\utilde{\omega}_\sigma(\widetilde{B}) = (\omega\otimes\sigma)(\Theta(\II\otimes B)).
\end{equation}  
 
 We now wish to identify a correspondence between probe observables and system observables, so that measurements on the probe may be interpreted as measurements of the system, conditioned by the preparation state $\sigma$. That is, to the probe observable $B\in\Bc(\Mb)$ we wish to identify a system observable $A\in\Ac(\Mb)$, depending on $B$ and $\sigma$, such that
\begin{equation}\label{eq:A}
\utilde{\omega}_\sigma(\widetilde{B}) = \omega(A)
\end{equation}
for all states $\omega$ of $\Ac(\Mb)$. This may be achieved as follows.  

Let $\eta_\sigma:\Ac(\Mb)\otimes\Bc(\Mb)\to \Ac(\Mb)$ be the map
extending $\eta_\sigma(A\otimes B)=\sigma(B)A$ by linearity (and continuity if appropriate),\footnote{The map $C\mapsto \eta_\sigma(C)\otimes \II\in \Ac(\Mb)\otimes\Bc(\Mb)$ is called the conditional expectation of $\Ac(\Mb)\otimes\Bc(\Mb)$ onto $\Ac(\Mb)\otimes\II$.}  
which thus obeys 
\begin{equation} 
A \eta_\sigma(C)=\eta_\sigma((A\otimes \II)C), \qquad\eta_\sigma(C)A=\eta_\sigma(C(A\otimes \II))
\end{equation} 
for $A\in\Ac(\Mb)$, $C\in\Ac(\Mb)\otimes\Bc(\Mb)$. The map $\eta_\sigma$ is completely positive: its tensor products with any finite-dimensional matrix identity map preserve positivity. For completeness, this statement is proved in Appendix~\ref{appx:cpetc}.

Defining $\varepsilon_\sigma:\Bc(\Mb)\to\Ac(\Mb)$ by 
\begin{equation}\label{eq:inducedobs}
\varepsilon_\sigma(B)=(\eta_\sigma\circ\Theta)(\II\otimes B),
\end{equation} 
we have, using~\eqref{eq:probe_exp} and the definitions,
\begin{equation}
\omega(\varepsilon_\sigma(B)) = \omega((\eta_\sigma\circ\Theta)(\II\otimes B))
= (\omega\otimes\sigma)(\Theta(\II\otimes B)) = \utilde{\omega}_\sigma(\widetilde{B}) ,
\end{equation}
which is the required identification between probe and system observables. 
Adapting the terminology of QMT, the probe theory $\Bc$, coupled theory $\Cc$, identification maps $\chi$ and probe preparation state $\sigma$
constitute a \emph{measurement scheme} for the induced system observable $\varepsilon_\sigma(B)$. Of course, nothing here actually requires that $B$ is self-adjoint, and we will
sometimes abuse terminology by referring to $\varepsilon_\sigma(B)$ as the \emph{induced system observable} corresponding to $B$ even when $B$ is not self-adjoint.
\begin{theorem}\label{thm:induced}
	For each probe preparation state $\sigma$, $A=\varepsilon_\sigma(B)$ is the unique solution to \eqref{eq:A} provided that $\Ac(\Mb)$ is separated by its states. In general, the map $\varepsilon_\sigma$ is a completely positive linear map and has the properties
	\begin{equation}\label{eq:cpprops}
	\varepsilon_\sigma(\II)=\II,\qquad \varepsilon_\sigma(B^*)=\varepsilon_\sigma(B)^*, \qquad
	\varepsilon_\sigma(B)^*\varepsilon_\sigma(B)\le \varepsilon_\sigma(B^*B) \,.
	\end{equation} 	
	For fixed $B$, the map $\sigma\mapsto \varepsilon_\sigma(B)$ is weak-$*$ continuous.
\end{theorem}
\begin{proof}
	The first statement summarises the foregoing discussion; we remark only that the vector states in (a common dense domain within) a faithful Hilbert space representation provide a separating set of states. Complete positivity and the properties listed in~\eqref{eq:cpprops} are proved in Appendix~\ref{appx:cpetc}, while weak-$*$ continuity follows from the definition of $\eta_\sigma$. 
\end{proof} 
In particular, self-adjoint elements of $\Bc(\Mb)$ are mapped to self-adjoint elements of $\Ac(\Mb)$.
However, $\varepsilon_\sigma$ is in general neither injective, nor an algebra homomorphism. We mention that all $C^*$-algebras are separated by their states;
indeed this applies whenever $\Ac(\Mb)$ admits a faithful Hilbert space representation, and in QFT it is typical that the \emph{physical} representations admit cyclic and separating vectors (see, e.g.~\cite{FewVerch_aqftincst:2015}). 

Theorem~\ref{thm:induced} has an important consequence. Recall that the experimenter actually measures the observable $\widetilde{B}$ in state $\smash{\utilde{\omega}_\sigma}$ of the combined system, where $B=B^*$ is a probe observable. 
The induced system observable has been selected to satisfy~\eqref{eq:A}; that is, so that the expected outcome of the actual measurement agrees with the expectation value of $A =\varepsilon_\sigma(B)$ in state $\omega$. Owing to~\eqref{eq:cpprops}, the variance of the actual measurement is always at least as great as that of the induced observable:
\begin{align}
\Var (\widetilde{B};\utilde{\omega}_\sigma) &= \utilde{\omega}_\sigma(\widetilde{B}^2)-\utilde{\omega}_\sigma(\widetilde{B})^2 =
\omega(\varepsilon_\sigma(B^2))-\omega(\varepsilon_\sigma(B))^2 \notag\\ 
&\ge \omega(A^2)-\omega(A)^2
=\Var(A;\omega),
\end{align}
using that fact that $\tau^+$ is a homomorphism, so $(\widetilde{B})^2=\widetilde{(B^2)}$ by~\eqref{eq:Btilde}. 
(Here we use the standard formula for the quantum mechanical variance, implicitly assuming that the outcomes are distributed according to a spectral measure for a representation of $\widetilde{B}$.)
Therefore the actual measurement is less sharp than the hypothetical measurement of
the induced observable in the system state. The additional variance derives from 
quantum fluctuations in the probe; we will see later how this can be quantified in 
a particular model. 

Let us note that the measurement scheme is not time-symmetric, because our use of the `in' and `out' regions enforces the idea that preparations are made early while measurements are made late. Again, we emphasise that the distinction between early and late does not refer to an observer's clock, or frame of reference, but just to the covariant delineation of $\Mb^\pm$.  

\paragraph{Effect-valued measures}
The foregoing description of observables can be further resolved, with a view to an underlying probability interpretation. One may consider maps ${\sf E}:\Xc\to \Af$, where $\Xc$ is a $\sigma$-algebra and $\Af$ is a $*$-algebra, so that ${\sf E}$ has the properties of a measure (finitely additive, in the purely $*$-algebraic setting) and takes its values in the effects of $\Af$: that is, the elements $A\in\Af$ such that $A$ and $\II-A$ are both positive. In particular, we demand that ${\sf E}(\Omega_\Xc)=\II$, where $\Omega_\Xc$ is the total space of the $\sigma$-algebra $\Xc$. Then one interprets the elements of $\Xc$ as potential outcomes of the  measurement and for each $X\in\Xc$, $\omega({\sf E}(X))$ is the probability that a value lying in $X$ is measured in state $\omega$. See~\cite[Ch.~9]{Busch_etal:quantum_measurement} for a discussion.\footnote{The utility of this definition depends on the existence of sufficiently many effects, which is not guaranteed in $*$-algebras, although it is in $C^*$-algebras.} Any map ${\sf E}$ of this type will be called an \emph{effect-valued measure} (EVM). Note, that
Ref.~\cite{Busch_etal:quantum_measurement} use the term `observable' for what we call an EVM; 
we have chosen not to do this because the understanding of `observable' as a self-adjoint algebra element is so ingrained in AQFT.

In our present setting any EVM ${\sf E}$ taking values in the effects of the probe induces a corresponding EVM $X\mapsto \varepsilon_\sigma({\sf E}(X))$ taking values in the effects of the system. Here we use the linearity and positivity preserving properties of $\varepsilon_\sigma$. Due to~\eqref{eq:cpprops}, one has
$\varepsilon_\sigma({\sf E}(X))^2 \le \varepsilon_\sigma({\sf E}(X)^2)$
so, even if ${\sf E}$ happens to be sharp, with ${\sf E}(X)^2={\sf E}(X)$, the   induced EVM $\varepsilon_\sigma\circ{\sf E}$ will generally not be sharp:
one knows only that $\varepsilon_\sigma({\sf E}(X))^2 \le \varepsilon_\sigma({\sf E}(X))$. This is another illustration of how probe fluctuations increase variance in measurement outcomes.
We see that the induced system observable is typically less sharp than the probe observable, but (as already noted) sharper than
the actual measurement made on the coupled system.

\paragraph{Localisation properties}  
Now suppose that $L$ is a (possibly disconnected) open causally convex subset of $M$ contained in $K^\perp$, so in particular $L\subset M^+\cap M^-$. 
We have already shown that $\Theta$ acts trivially on $\Ac(\Mb;L)\otimes \Bc(\Mb;L)$, and now use this fact to make two simple but important observations.

First, let $B\in\Bc(\Mb;L)$ be a probe observable localisable in $L$. Then $\Theta$ leaves  $\II\otimes B$ invariant and hence
\begin{equation}
\varepsilon_\sigma(B)=\eta_\sigma(\Theta (\II\otimes B)) = \eta_\sigma(\II\otimes B) = 
\sigma(B)\II.
\end{equation} 
This shows that the system observable induced by a probe observable belonging to the causal complement of the coupling region is a fixed
multiple of the identity (determined by the probe preparation state and $B$) and provides no information about the field. 

Second, suppose that $A\in \Ac(\Mb;L)$ is a system observable localised in $L$ and let $B\in\Bc(\Mb)$ be any probe observable.
Then we may compute 
\begin{align}
[\varepsilon_\sigma(B) ,A] &= [\eta_\sigma(\Theta (\II\otimes B)),A] 
= \eta_\sigma([\Theta (\II\otimes B),A\otimes \II)]) \notag\\ 
&= \eta_\sigma(\Theta [\II\otimes B,A\otimes \II]) 
=0,
\end{align}
as $A\otimes \II$ is invariant under $\Theta$. This shows that the induced
observable $\varepsilon_\sigma(B)$ commutes with all system observables localised in the 
causal complement of $K$. Therefore, all field observables induced by probe observables are localisable in any connected open causally convex set containing the coupling region $K$.
Whether or not it is possible to provide tighter localisation information will be discussed later in the context of a specific model.
Some general considerations of localization concepts for observables in quantum field theory on Minkowski spacetime appear in \cite{Kuckert:2000}, in a model-independent context.

The results of this discussion may be summarised as follows.
\begin{theorem}\label{thm:localisation}
	For each probe observable $B\in\Bc(\Mb)$, the induced system observable $\varepsilon_\sigma(B)$ may be localised in any connected open causally convex set containing $K$. If $B$ may be localised in $K^\perp$ then $\varepsilon_\sigma(B)=\sigma(B)\II$.
\end{theorem}
Note that any causally convex set containing $K$ also contains the causal hull of $K$.

\paragraph{Joint EVMs} We have seen that the induced observables may be localised  
in any suitable (i.e., open, connected and causally convex) neighbourhood of the causal hull $J^+(K)\cap J^-(K)$ of the coupling region $K$. This raises the following question. Consider two
experimenters in causally disjoint spacetime regions. Each can measure a local observable of the probe and Einstein causality entails that these observables commute and are therefore compatible. However, the corresponding induced system observables may both be localised in some suitable neighbourhood of the causal hull of $K$ and (as there is no reason to suppose they have causally disjoint localisation) may be incompatible as system observables. How is this to be reconciled with the compatibility of the probe observables?

The answer may be given using the notion of a \emph{joint EVM}. Here,
two EVMs ${\sf E}_i:\Xc_i\to \Af$ are said to have a \emph{joint EVM} if 
they are the marginals of an EVM ${\sf E}:\Xc_1\otimes\Xc_2\to \Af$, that is, ${\sf E}_1(X_1)={\sf E}(X_1\times \Omega_{\Xc_2})$, ${\sf E}_2(X_2)={\sf E}(\Omega_{\Xc_1}\times X_2)$ (see~\cite[Ch.~11]{Busch_etal:quantum_measurement}). If ${\sf E}$ is a joint EVM valued in the effects of the probe, then it is obvious that $\varepsilon_\sigma\circ {\sf E}$ is a joint EVM for the system EVMs $\varepsilon_\sigma\circ {\sf E}_i$.
However, even if the ${\sf E}_i$ are commuting EVMs (perhaps because they relate to causally disjoint probe measurements) it is not generally the case that the induced system EVMs will commute. This is simply because $\varepsilon_\sigma$ is not generally a homomorphism. In the same vein, 
even if the EVMs ${\sf E}_i$ are sharp, i.e., projection-valued, the same will not necessarily be true of the $\varepsilon_\sigma\circ {\sf E}_i$. 

The importance of these simple observations is that on the one hand, they protect
the freedom of experimenters in causally disjoint regions to independently measure the probe in any way they wish (in particular, sharply or unsharply), while on the other, they protect the principle that incompatible system observables cannot be measured jointly and sharply. The resolution is that the information the experimenters can obtain concerning incompatible system observables is limited to what can be provided by an unsharp joint EVM.

\subsection{Instruments and successive measurements}\label{sec:instruments}

\paragraph{Instruments}

Suppose a probe-effect $B$ is observed. We would like to obtain a new system state that is conditioned on the observation of this effect, which means that the
new state correctly predicts the conditional probability for the joint observation of $B$ together with any system effect, given that $B$ is observed.

Let $A$ and $B$ be effects of the system and probe, respectively, and consider a joint measurement of $A$ and $B$ at late times. By the same reasoning as used above, the probability of the joint effect being observed is
\begin{equation}
\Prob_\sigma(A\& B;\omega) = \omega(\eta_\sigma\Theta (A\otimes B))
\end{equation}
and the conditional probability that $A$ is observed, given that $B$ is observed, is
\begin{equation}\label{eq:conditionalprob}
\Prob_\sigma(A|B;\omega) = \frac{\Prob_\sigma(A\& B;\omega)}{\Prob_\sigma(B;\omega)} = \frac{(\If_\sigma(B)(\omega))(A)}{(\If_\sigma(B)(\omega))(\II)},
\end{equation}
where we write (for any $A\in\Ac(\Mb)$)
\begin{equation}\label{eq:instrument}
(\If_\sigma(B)(\omega))(A):= (\omega\otimes\sigma)(\Theta(A\otimes B)) = (\Theta^*(\omega\otimes\sigma))(A\otimes B).
\end{equation}
In particular, the normalising factor is $(\If_\sigma(B)(\omega))(\II)=\omega(\varepsilon_\sigma(B))$.
We call the map $\If_\sigma(B):\Ac(\Mb)^*_+\to \Ac(\Mb)^*_+$  the \emph{pre-instrument} corresponding to effect $B$ and probe preparation state $\sigma$. The relation~\eqref{eq:conditionalprob} justifies the interpretation that 
$\If_\sigma(B)(\omega)$ is the unnormalized updated system state conditioned on the probe effect $B$ being observed. Our argument here has followed that of~\cite[\S 10.2]{Busch_etal:quantum_measurement} while also adapting it to our present context.
The normalized state  
\begin{equation}
\omega':=\frac{\If_\sigma(B)(\omega)}{\If_\sigma(B)(\omega)(\II)} 
\end{equation} 
is the \emph{post-selected system state} after selective measurement of the probe. It is obvious that $\omega'$ is normalized; to check that $\omega'$ is indeed positive,
recall that the effect $B$ is positive and therefore takes the form $B=\sum_i C_i^*C_i$ for some finite set of elements $C_i\in\Bc(\Mb)$ (in a $C^*$-algebraic setting one can just write $B=C^*C$ for $C=B^{1/2}$). Then
\begin{equation}
\If_\sigma(B)(\omega)(A^*A) = \sum_i (\Theta^*(\omega\otimes\sigma))((A\otimes C_i)^*(A\otimes C_i))\ge 0
\end{equation} 
and it follows that $\omega'(A^*A)\ge 0$ for all $A\in\Ac(\Mb)$.

If one is given an EVM ${\sf E}:\Xc\to \Bc(\Mb)$ then the composition of the pre-instrument with ${\sf E}$ gives a \emph{instrument} in the sense originally introduced by Davies and Lewis~\cite{DaviesLewis:1970}, i.e., a measure $X\mapsto \If_\sigma({\sf E}(X))$ on the $\sigma$-algebra of measurement outcomes valued in positive maps on the state space. In fact this would even be a CP-instrument but we will usually drop the prefix.

\paragraph{Non-selective measurement} In a \emph{non-selective probe measurement}, there is no filtering conditional on the measurement outcome. Using the preceding definitions, a non-selective probe measurement of an EVM ${\sf E}:\Xc\to \Bc(\Mb)$ corresponds to the pre-instrument $\If_\sigma({\sf E}(\Omega_\Xc))=\If_\sigma(\II)$.
A justification for this definition is easily given in the case that $\Omega_\Xc$ is a finite set, for then the additivity properties of the instrument require that
\begin{equation}
\sum_{a\in \Omega_\Xc} \If_\sigma({\sf E}(\{a\})) = \If_\sigma({\sf E}(\Omega_\Xc))
\end{equation}
while the left-hand side, evaluated on $\omega$, is clearly the sum of all the updated states for each possible outcome $a\in\Omega_\Xc$, weighted by their respective probabilities. The same result may be obtained using any other partition of $\Omega_\Xc$.

Explicitly, the updated state resulting from the non-selective measurement is
\begin{equation}
\omega'_{\text{ns}}(A) =  \If_\sigma(\II)(\omega)(A) = (\Theta^*(\omega\otimes\sigma))(A\otimes \II).
\end{equation}
In other words, $\omega'_{\text{ns}}$ is the partial trace of the state $
\Theta^*(\omega\otimes\sigma)$ over the probe. It depends only on the dynamics of the coupling and not on the EVM ${\sf E}$ that was being non-selectively measured. 
In particular, if $A\in\Ac(\Mb)$ may be localised in the causal complement of the coupling region then $\Theta (A\otimes \II)=A\otimes \II$, so $\omega'_{\text{ns}}(A) = \omega(A)$. Just as it should be, 
a non-selective measurement cannot influence the results of other experiments in causally disjoint regions. This is not the case in selective measurement, as we now show.

\paragraph{Locality and post-selection}  
Now let $A$ be a system observable localisable in the causal complement of $K$ and let $B$ be a probe effect, without assumptions on its localisation. 
Using again the fact that $\Theta (A\otimes \II)=A\otimes \II$, and
noting that
\begin{equation}
A\varepsilon_\sigma(B)
= A \eta_\sigma \Theta(\II\otimes B)  = \eta_\sigma\left( (A\otimes \II) \Theta (\II\otimes B)\right)=
\eta_\sigma(\Theta (A\otimes B)),
\end{equation}
the definition of the pre-instrument in~\eqref{eq:instrument} is
\begin{equation}
\If_\sigma(B)(\omega)(A) = \omega(\eta_\sigma(\Theta (A\otimes B)))=\omega(A\varepsilon_\sigma(B)).
\end{equation}
Accordingly, the normalized post-selected state, conditioned on the effect being observed, is
\begin{equation}\label{eq:postselected}
\omega'(A) = \frac{\omega(A\varepsilon_\sigma(B))}{\omega(\varepsilon_\sigma(B))}
\end{equation}
for system observables $A$ localisable in $K^\perp$. 

The following result now follows easily: 
\begin{theorem}
	Consider a measurement of a probe effect $B$ in which the effect is observed. For each $A\in\Ac(\Mb;K^\perp)$, the expectation value of $A$ is unchanged in the post-selected state $\omega'$ if and only if $A$ is uncorrelated with $\varepsilon_\sigma(B)$ in the original system state $\omega$.
\end{theorem}
\begin{proof}
	Evidently $\omega'(A)=\omega(A)$ holds if and only if $\omega(A\varepsilon_\sigma(B)) = \omega(A)\omega(\varepsilon_\sigma(B))$. 
\end{proof}

The interpretation of $\omega'$ requires some care. By construction, it is the result of applying post-selection to the state $\omega$, conditioned on the results of the measurement of probe observable $B$. In general, the post-selected state assigns different expectation values to $\omega$ to observables localised in any region of spacetime: even those in the causal past or causal complement of the interaction region. Of course, this does not change events that have happened; the point is simply that the probabilities for those events are different in the post-selected state conditioned on a particular measurement outcome. The reason that probabilities can change for observables localised in the causal complement is simply one of correlation. One might think of a spin measurement made on one member of an entangled pair of spins in a singlet state. Conditioned on the 
result of that measurement, the result of a measurement on the remaining spin may be predicted with certainty, even if the two measurements are causally disjoint. 

It is also instructive to consider a situation of high correlation. Suppose that $\omega$ has the Reeh--Schlieder property (see e.g.,~\cite[\S 4.5.4]{FewVerch_aqftincst:2015}). Then the assumption that $\omega'$ agrees with $\omega$ on observables localised in some $O$ within the causal complement of $K$ entails that
\begin{equation}
\omega\left(A_1^*\left(\frac{\varepsilon_\sigma(B)}{\omega(\varepsilon_\sigma(B))}-\II\right) A_2\right)= \omega'(A_1^*A_2) - \omega(A_1^*A_2)
=0
\end{equation}
for all $A_i\in\Ac(\Mb;O)$, using the commutativity of $A_2$ and $\varepsilon_\sigma(B)$. By the Reeh--Schlieder property of $\omega$ this then implies that the induced system observable $\varepsilon_\sigma(B)$ is a multiple of the unit. Therefore, any nontrivial probe measurement necessarily alters the state in the causal complement of the coupling region in this highly correlated case. As is well known, the Reeh--Schlieder property encodes a strong form of quantum entanglement between causally disjoint regions and our argument shows that causally disjoint probe measurements can in principle be used to detect EPR-like correlations~\cite{SummersWerner:1987,SummersWerner:1995,CliftonHalvorson,VerchWerner:2005,HollandsSanders:2017}. 

One might wonder whether $\omega'$ might agree with $\omega$ in portions of the region
	$J^-(K)\setminus K$ to the past of the coupling region. To see that this is unlikely, suppose that the system obeys Huygens' principle and suppose that $O\subset\Mb$ has no null geodesics connecting it to $K$. In that case $\Theta$ will act trivially on $A\otimes\II$ for all $A\in\Ac(\Mb;O)$ 
	and formula~\eqref{eq:postselected} will hold, as will the Reeh--Schlieder argument just made. Thus, in general, there seems no reason to assume that $\omega'$ agrees with $\omega$ in the backward causal cone of $K$.

Given that $\omega'$ does not agree with $\omega$ in any geometrically determined region of spacetime, there seems to be no purpose in envisaging a transition from $\omega$ to $\omega'$ occurring along or near some
surface in spacetime (whether a constant time surface as in non-relativistically inspired accounts of measurement, or e.g., along the backward light cone of the interaction region as in the proposal of Hellwig and Kraus~\cite{HellwigKraus_prd:1970}, or an earlier proposal of Schlieder~\cite{Schlieder:1968}). In fact, as Hellwig and Kraus recognised, whether the state actually remains unchanged or not in the past of the coupling region is a `pure convention' with no operational significance as the region is no longer accessible to further experiment. It is also of no consequence at what stage the probe itself is measured; this may take place far from the coupling region both in space and time. 

In view of these considerations, there seems no reason to invoke a physical process of  state-reduction occurring at points or surfaces in spacetime, rather, the updated state reflects the observer's filtering of the system by conditioning on measurement outcomes. Having said that, we regard the important message of the paper of Hellwig and Kraus to be that multiple measurements can be made in a consistent way within QFT in Minkowski spacetime. We will now proceed to show that this is true also in our framework, under specified assumptions, even in curved spacetimes.

\paragraph{Successive measurements}
Now suppose that two measurements of the field system are made in compact interaction regions $K_1$ and $K_2$. We suppose that $K_2\cap J^-(K_1)$ is empty, so that $K_2$ may be regarded as later than $K_1$ at least by some observers. Later, we will consider the situation in which they are causally disjoint and can be ordered in either way.

We consider two probe systems $\Bc_i(\Mb)$ and coupled systems $\Cc_i(\Mb)$,
with coupling regions $K_i$, each corresponding to its own scattering morphism $\Theta_i$ on 
$\Ac(\Mb)\otimes\Bc_i(\Mb)$. On the three-fold tensor product $\Ac(\Mb)\otimes\Bc_1(\Mb)\otimes\Bc_2(\Mb)$, we define $\hat{\Theta}_1=\Theta_1\otimes_3 \id_{\Bc_2(\Mb)}$, and  $\hat{\Theta}_2=\Theta_2\otimes_2 \id_{\Bc_1(\Mb)}$, where the subscript on the tensor product indicates the slot into which the second factor is inserted. Taken together, the two probes may be considered as a single probe with algebra $\Bc_1(\Mb)\otimes\Bc_2(\Mb)$ and coupling region $K_1\cup K_2$ and a combined
scattering morphism $\hat{\Theta}$ on $\Ac(\Mb)\otimes\Bc_1(\Mb)\otimes\Bc_2(\Mb)$.
We assume that the scattering morphism obeys a natural \emph{causal factorisation formula}
\begin{equation}\label{eq:Bogoliubov}
\hat{\Theta} = \hat{\Theta}_1\circ\hat{\Theta}_2,
\end{equation}
(related to a special case of Bogoliubov's factorisation formula in perturbative AQFT~\cite{Rejzner_book,DuetschFredenhagen:2000,BogoliubovShirkov})
recalling that our scattering morphism maps observables from the future to the past.
The main result of this section is that the instruments corresponding to the individual and combined measurements act in a coherent fashion. 

\begin{theorem}\label{thm:causalcomp}
Consider two probes as described above, with $K_2\cap J^-(K_1)=\emptyset$. For all probe preparations $\sigma_i$ of $\Bc_i(\Mb)$ and all probe observables $B_i\in\Bc_i(\Mb)$, the following identity for the pre-instruments holds:
\begin{equation}\label{eq:composite_instrument1}
	\If_{\sigma_2}(B_2)\circ \If_{\sigma_1}(B_1) = \If_{\sigma_1\otimes \sigma_2}(B_1\otimes B_2).
\end{equation}
If, in fact, $K_1$ and $K_2$ are causally disjoint, we have
\begin{equation}\label{eq:composite_instrument2}
\If_{\sigma_2}(B_2)\circ \If_{\sigma_1}(B_1) = \If_{\sigma_1\otimes \sigma_2}(B_1\otimes B_2)  = \If_{\sigma_1}(B_1)\circ \If_{\sigma_2}(B_2).
\end{equation}
\end{theorem}
This is a key result that permits experiments to be analysed into their causally constituent parts;
as, for example, in the discussion of impossible measurements~\cite{BostelmannFewsterRuep:2020}. From a mathematical perspective it shows that there is a monoidal structure on pre-instruments which is even symmetric for causally disjoint coupling regions. In general, however, if the couplings are strictly causally ordered, the 
symmetry is broken. This would not be seen in a Euclidean QFT framework.
\begin{proof}
	The composition of the pre-instruments is computed as follows. For any system state $\omega$ and system observable $A$, we have 
	\begin{align}
	\If_{\sigma_2}(B_2)(\If_{\sigma_1}(B_1)(\omega))(A) &= (\If_{\sigma_1}(B_1)(\omega)\otimes\sigma_2)(\Theta_2 (A\otimes B_2))\notag \\
	&= \If_{\sigma_1}(B_1)(\omega)(\eta_{\sigma_2}(\Theta_2(A\otimes B_2))) \notag\\
	&= (\Theta_1^*(\omega\otimes\sigma_1))(\eta_{\sigma_2}(\Theta_2(A\otimes B_2))\otimes B_1) \notag\\
	&=(\Theta_1^*(\omega\otimes\sigma_1)\otimes_2\sigma_2)(\Theta_2(A\otimes B_2))\otimes_3 B_1)
	\end{align}
	where we have decorated the tensor products with subscripts where necessary to indicate the slot into which the second factor is inserted. Now we are free to permute the second and third tensor factors if we do so in both the algebras and the functionals acting thereon. Therefore
	\begin{align}
	\If_{\sigma_2}(B_2)(\If_{\sigma_1}(B_1)(\omega))(A)
	 &=(\Theta_1^*(\omega\otimes\sigma_1)\otimes_3\sigma_2)(\Theta_2(A\otimes B_2))\otimes_2 B_1) \notag\\
	&= (\omega\otimes\sigma_1\otimes\sigma_2)((\Theta_1\otimes_3\id)(\Theta_2(A\otimes B_2)\otimes_2 B_1))\notag\\
	&= (\omega\otimes\sigma_1\otimes\sigma_2)((\hat{\Theta}_1\circ\hat{\Theta}_2)(A\otimes B_1\otimes B_2))\notag\\	
	&= (\omega\otimes\sigma_1\otimes\sigma_2)(\hat{\Theta} (A\otimes B_1\otimes B_2)) \notag\\
	&= \If_{\sigma_1\otimes \sigma_2}(B_1\otimes B_2)(\omega)(A)
	\end{align}
	where we have used the causal factorisation formula~\eqref{eq:Bogoliubov}. This proves the first statement; the second is an immediate consequence, as~\eqref{eq:Bogoliubov} now implies
	$\hat{\Theta}_2\circ\hat{\Theta}_1=\hat{\Theta}$.
\end{proof}

\begin{corollary}
	Consider two probes as described above, with $K_2\cap J^-(K_1)=\emptyset$, effects $B_i\in\Bc_i(\Mb)$ and probe preparation states $\sigma_i$ ($i=1,2$). Suppose $B_1$ has nonzero probability of being observed in system state $\omega$, and that $B_2$ has nonzero probability of being observed in system state $\omega'_1$, the post-selected system state conditioned on $B_1$ being observed in state $\omega$. Then the post-selected state $\omega''_{12}$ conditioned on $B_2$ being observed in state $\omega'_1$ coincides with the post-selected state $\omega'_{12}$ conditioned on $B_1\otimes B_2$ being observed in state $\omega$.
\end{corollary}
\begin{proof}
	We may compute 
	\begin{equation}
	\omega'_1 = \frac{\If_{\sigma_1}(B_1)(\omega)}{\If_{\sigma_1}(B_1)(\omega)(\II)}
	\end{equation}
	and 
	\begin{equation}
	\omega''_{12} = \frac{\If_{\sigma_2}(B_2)(\omega'_1)}{\If_{\sigma_2}(B_2)(\omega_1')(\II)}
	\end{equation}
	conditioned on both effects being observed, post-selecting on the $B_1$ measurement at the intermediate step. Obviously, the normalisation factors applied to $\omega_1'$ cancel in the formula for $\omega''_{12}$ and so we also have
	\begin{equation}
	\omega''_{12} = \frac{\If_{\sigma_2}(B_2)(\If_{\sigma_1}(B_1)(\omega))}{\If_{\sigma_2}(B_2)(\If_{\sigma_1}(B_1)(\omega))(\II)} = 
	\frac{\If_{\sigma_1\otimes \sigma_2}(B_1\otimes B_2)(\omega)}{\If_{\sigma_1\otimes \sigma_2}(B_1\otimes B_2)(\omega)(\II)} = \omega'_{12}
	\end{equation}
	where the denominators are equal by setting $A=\II$ in \eqref{eq:composite_instrument1}. 
\end{proof}
If the $K_i$ are causally disjoint then the post-selection may be made in either order, with the same result. We emphasise that we have not needed to invoke any reduction of the state across geometric boundaries in spacetime; everything follows from the basic definitions that we have set out, and from the causal factorisation formula. The latter must be verified for concrete models of system-probe interactions. 

Exactly as one would hope, we have shown that experiments conducted in causally disjoint regions [that is, for which the coupling regions are causally disjoint] may be conducted in ignorance of one another, or combined as a single overall experiment by coordinating their results.

\paragraph{Locally performed operations} A post-selected effect measurement partly coincides with the idea of a locally performed operation. 
It is convenient to switch to a $C^*$-algebraic setting for these purposes. Fix a probe preparation state, let $B$ be a probe effect as before with nonzero probability
of being observed in system state $\omega$, and let $\omega'$ be the post-selected state.
 Using~\eqref{eq:postselected} and the positivity-preserving property of $\varepsilon_\sigma$, we see that
\begin{equation}
\omega'(A)= \frac{\omega(A \varepsilon_\sigma(B))}{\omega(\varepsilon_\sigma(B))} = 
\frac{\omega(\varepsilon_\sigma(B)^{1/2}A \varepsilon_\sigma(B)^{1/2})}{\omega(\varepsilon_\sigma(B))}  
\end{equation} 
for any $A\in\Ac(\Mb;N)$ with $N\subset K^\perp$. Here we have used the fact that
square roots can be taken inside the local algebras. For observables localisable in the causal complement of $K$, the state appears to have been produced by an operation $\varepsilon_\sigma(B)^{1/2}$ which is localisable in 
any connected causally convex neighbourhood of the causal hull of $K$.

\subsection{Symmetries}\label{sec:symm}

\paragraph{Gauge invariance} It may be that not all elements of the algebra describing a theory should be regarded as observable. This is the case where a global gauge symmetry exists; only those elements that are gauge-invariant are to be regarded as potentially observable (unless the measurement coupling breaks the symmetry, for instance). This can be easily accommodated within our general scheme as we now show.

Let $G$ be a common group of global gauge transformations for $\Ac$, $\Bc$ and $\Cc$.
That is, to each open causally convex region $L$ of $\Mb$ (including the possibility $L=M$) there is an action $\varphi_\Lb$ of $G$ on $\Ac(\Lb)$ such that
\begin{equation}\label{eq:gauge}
\alpha_{\Lb;\Lb'}\circ \varphi_{\Lb'}(g) = \varphi_\Lb(g)\circ\alpha_{\Lb;\Lb'}, \qquad g\in G
\end{equation}
holds for each pair of such regions with $L'\subset L$. This definition is motivated by the
discussion of global gauge groups in the context of locally covariant QFT, where they arise as functorial automorphism groups~\cite{Fewster:gauge}. Similarly, there should be a comparable actions $\psi$ on $\Bc$ and $\xi$ on $\Cc$. 
We say that the coupling is gauge-invariant if these actions are related by 
\begin{equation}\label{eq:gaugeinv}
\xi_\Lb(g)\circ\chi_\Lb = \chi_\Lb\circ (\varphi_\Lb(g)\otimes\psi_\Lb(g))
\end{equation} 
for any $L$ contained in the causal complement of $K$,
recalling that $\chi_\Lb$ is the isomorphism from $\Ac(\Lb)\otimes\Bc(\Lb)$ to $\Cc(\Lb)$ that exists in this situation. Combining~\eqref{eq:gaugeinv} and~\eqref{eq:naturality}, we deduce
\begin{align}
\xi(g)\circ\kappa^\pm &= \xi(g)\circ\gamma^\pm \circ \chi^\pm = \gamma^\pm\circ\xi^\pm(g)\circ \chi^\pm = \gamma^\pm\circ\chi^\pm \circ (\varphi^\pm(g)\otimes\psi^\pm(g)) \notag \\ 
&= \kappa^\pm\circ (\varphi^\pm(g)\otimes\psi^\pm(g)) 
\end{align}
(as usual we abbreviate $\varphi_\Mb$ and $\varphi_{\Mb^\pm}$ to $\varphi$ and $\varphi^\pm$ etc): and hence 
\begin{equation}
\xi(g)\circ\tau^\pm = \tau^\pm\circ (\varphi(g)\otimes\psi(g)).
\end{equation}
Consequently, the scattering transformation is gauge-invariant,
\begin{equation}\label{eq:theta_inv}
(\varphi_\Mb(g)\otimes\psi_\Mb(g))\circ\Theta = (\tau^-)^{-1}\circ\xi(g)\circ\tau^+= \Theta\circ(\varphi_\Mb(g)\otimes\psi_\Mb(g)).
\end{equation}
Under these circumstances, the induced observables transform in an equivariant fashion.
\begin{theorem}
	The induced observables obey
	\begin{equation}
	\varphi(g)(\varepsilon_\sigma(B)) = \varepsilon_{\psi(g^{-1})^*\sigma}(\psi(g)B).
	\end{equation}
	In particular, if $\sigma$ is gauge-invariant, then gauge-invariant probe observables induce gauge-invariant system observables.
\end{theorem}
\begin{proof}
	First, note that 
	\begin{align}
	\varphi(g)(\eta_\sigma(A\otimes B)) &= \sigma(B)\varphi(g)A = 
	(\psi(g^{-1})^*\sigma)(\psi(g)(B))\varphi(g)A \notag\\
	&=
	\eta_{\psi(g^{-1})^*\sigma} ((\varphi(g)\otimes\psi(g))(A\otimes B)).
	\end{align}
	so $\varphi(g)\circ\eta_\sigma = \eta_{\psi(g^{-1})^*\sigma}\circ (\varphi(g)\otimes\psi(g))$. Using the definition~\eqref{eq:inducedobs} and the gauge-invariance~\eqref{eq:theta_inv} of the scattering morphism, the first statement is proved, and the second is immediate.
\end{proof}
This result shows how the 
(un)observability of gauge (non)invariant quantities is passed from the probe to the system. Of course this connection is removed if the coupling is not gauge invariant in the above sense, or if the probe preparation state is not gauge-invariant. For example, the average magnetisation of a spin chain is not gauge-invariant under simultaneous rotation of all spins, but can be measured if one couples to a fixed external field that breaks this invariance.\footnote{We thank J\"{u}rg Fr\"{o}hlich for stimulating this remark.}

\paragraph{Geometrical symmetries} 
Suppose that $\Mb$ admits a time-translation symmetry that is represented in the system theory by a $1$-parameter group of automorphisms $\nu_t$ of $\Ac(\Mb)$,
so that $\nu_t\Ac(\Mb;N)=\Ac(\Mb;N_t)$, where $N_t$ is a translation of $N$ to the future if $t>0$.  We will say that the state $\omega$ satisfies \emph{strong mixing} with respect to these translations if it obeys the timelike clustering condition
\begin{equation}
\omega((\nu_s A)B)- \omega(\nu_s A)\omega(B)\longrightarrow 0, \qquad (|s|\to\infty)
\end{equation}
for each fixed $A,B\in\Ac(\Mb)$. Strong mixing plays an important role in discussions of stability of KMS and ground states~\cite{BratRob,BratteliRobinson} and is known to hold for ground states in typical QFTs on Minkowski space~\cite{Maison:1968}. 

Clearly, if $\omega$ satisfies future-asymptotic clustering, then
the post-selected state conditioned on the observation of probe effect $B$ obeys
\begin{equation}\label{eq:return}
\omega'(\nu_s A) - \omega(\nu_s A) =  \frac{\omega((\nu_s A)\varepsilon_\sigma(B))}{\omega(\varepsilon_\sigma(B))}- \omega(\nu_s A) \longrightarrow 0 \qquad (|s|\to \infty).
\end{equation}
Thus, the distinction between the post-selected and original states is erased to arbitrary accuracy by waiting long enough. Consequently, experiments may be repeated without active re-preparation of the system state, if it is strongly mixing; similarly, the post-selected state is also well-approximated by $\omega$ sufficiently to the past of the coupling region.

More generally, suppose that $\Mb$ admits a time-orientation preserving isometry $\psi$, 
which is a symmetry of the theories $\Ac$ and $\Bc$, implemented by automorphisms $\alpha_\psi$ and $\beta_\psi$ of $\Ac(\Mb)$ and $\Bc(\Mb)$. Then any coupling $\Cc$ of $\Ac$ and $\Bc$ within a compact coupling region $K\subset M$ and scattering morphism $\Theta$ can be translated forwards by $\psi$ to give a modified coupling $\Cc'$ with coupling region $\psi(K)$ and scattering morphism $\Theta'=(\alpha_\psi\otimes\beta_\psi)\circ \Theta\circ (\alpha_\psi\otimes\beta_\psi)^{-1}$. 
Details are left to the reader; however, it is easily verified that our formalism is covariant
in the sense that
\begin{equation}
\varepsilon'_{\sigma'}(\beta_\psi B)=\alpha_\psi (\varepsilon_{\beta_\psi^*\sigma'}(B)).
\end{equation}

\section{A specific probe model}\label{sec:probe}

In this section we present a simple model of system-probe interaction in QFT, which is fully rigorous and explicitly solvable. Both the system and the probe are modelled by free scalar fields of possibly nonzero mass. They will be coupled linearly together in a bounded region of spacetime. The quantization of these systems (at least for sufficiently weak coupling) can follow standard lines and we will be fairly brief. See~\cite{FewVer:20xx,Baer:2015,AdvAQFT} for more details.

\paragraph{Classical action}
The combined action for the uncoupled systems on the spacetime $\Mb$ is
\begin{equation}
S_0 = \frac{1}{2}\int_M \dvol \left((\nabla_a \Phi)(\nabla^a\Phi) - m_\Phi^2 \Phi^2 +  (\nabla_b \Psi)(\nabla^b\Psi) - m_\Psi^2 \Psi^2\right),
\end{equation}
where $m_\Phi$ and $m_\Psi$ are the masses of the two fields. One could easily add
couplings to the curvature scalar but we omit this for simplicity.
The uncoupled field equations are thus
\begin{equation}
P\Phi = 0, \qquad Q\Psi= 0,
\end{equation}
where $P=\Box_\Mb+m_\Phi^2$ and $Q=\Box_\Mb+ m_\Psi^2$ are the Klein--Gordon operators. Standard theory~\cite{BarGinouxPfaffle} shows that $P$ and $Q$ have unique
advanced ($-$) and retarded ($+$) Green operators $E^\pm_P$ (and similarly for $Q$),
that is, continuous linear maps $E_P^\pm:\CoinX{M}\to C^\infty(M)$ obeying
\begin{equation}
E_P^\pm P f = f, \qquad P E_P^\pm f=f, \qquad  \supp E_P^\pm f\subset J^\pm(\supp f)
\end{equation}
for all $f\in\CoinX{M}$. Defining $E_P=E_P^--E_P^+$, 
every smooth solution $\Phi$ to $P\Phi=0$ with spatially compact support (i.e., support intersecting spacelike Cauchy surfaces compactly) can be expressed
in the form $\Phi=E_P f$ for some $f\in\CoinX{M}$. One can find an explicit
formula for a suitable $f$: let $\chi$ be a smooth function taking the value $1$ to the past of one Cauchy surface and vanishing to the future of another. Then 
$f=P\chi\Phi$ solves $E_P f=\Phi$ and is supported in the intersection of $\supp\Phi$ and the region between the Cauchy surfaces. Put another way, the identity $E_P=E_P P\chi E_P$ holds on $\CoinX{M}$. The possibility of generating any solution from a test function localised in an arbitrarily small neighbourhood of any Cauchy surface is related to the time-slice property. These above properties are common to a general class of \emph{Green-hyperbolic operators}~\cite{Baer:2015}. 

The systems will be coupled together via an interaction term 
\begin{equation}
S_{\text{int}}= - \int_\Mb\dvol\, \rho \Phi  \Psi,
\end{equation}
where $\rho$ is a real, smooth function with support contained in a compact set $K$. The field equation for the coupled system, with action $S=S_0 + S_{\text{int}}$, is
\begin{equation}\label{eq:eom}
\begin{pmatrix}
	P & R \\ R & Q
\end{pmatrix}\begin{pmatrix}
	\Phi\\ \Psi
\end{pmatrix}
=0,
\end{equation}
where $R\Phi=\rho\Phi$.  It will be convenient to write the matrix of operators as $T$ and to combine the fields as a function $\Xi\in C^\infty(M;\CC^2)$, writing the equation of motion as
\begin{equation}
T\Xi =0.
\end{equation}
The operator $T$ is also Green-hyperbolic, at least for sufficiently weak coupling~\cite{FewVer:20xx}, and so has advanced and retarded Green functions with analogous properties to those of $P$ and $Q$. As before, we will define the `in' ($-$) and `out' ($+$) regions by $M^\pm = M\setminus J^\mp(K)$. 

\paragraph{Quantization}
The quantization of the uncoupled systems is standard, and essentially the same methods can be used to deal with the coupled system. Let us start with the uncoupled field $\Phi$, which will be our `system'. The quantum field theory is described using
a $*$-algebra of observables $\Ac(\Mb)$ that is generated by a unit element together with other generators labelled $\Phi(f)$ ($f\in\CoinX{M}$), subject to the 
relations
\begin{itemize}
	\item[Q1] $f\mapsto \Phi(f)$ is $\CC$-linear
	\item[Q2] $\Phi(\overline{f})=\Phi(f)^*$ for all $f\in\CoinX{M}$
	\item[Q3] $\Phi(Pf)=0$ for all $f\in\CoinX{M}$
	\item[Q4] $[\Phi(f),\Phi(h)]=i E_P(f,h)\II$ for all $f,h\in\CoinX{M}$, where
	\begin{equation}
	E_P(f,h):=\int_M \dvol \, f E_Ph.
	\end{equation}
\end{itemize} 
The algebra $\Ac(\Mb)$ is known to be nontrivial and simple (see, e.g., \cite[\S 5.1]{FewVer:dynloc2}). Any Green hyperbolic operator may be quantized in the same way,
which applies in particular to the probe system (corresponding to $Q$), the uncoupled combined system-probe ($P\oplus Q$) and the coupled system ($T$). The latter two have generators labelled by test functions in $\CoinX{M;\CC^2}\cong \CoinX{M}\oplus\CoinX{M}$. 

The resulting QFT (for any Green-hyperbolic operator) obeys the general properties set out in Section~\ref{sec:prelim}. If $N$ is an open causally convex subset of $\Mb$, we define $\Ac(\Nb)$ following the above prescription for the region $N$ with the metric and causal structures induced from $\Mb$. Then the formula
\begin{equation}
\alpha_{\Mb;\Nb}\Phi_\Nb(f) = \Phi_\Mb(f),\qquad f\in\CoinX{N}
\end{equation}
extends to a morphism from $\Ac(\Nb)$ to $\Ac(\Mb)$ -- here we have temporarily decorated the smeared fields with a subscript to indicate whether they are elements of $\Ac(\Mb)$ or $\Ac(\Nb)$. The compatibility requirement~\eqref{eq:functorial} is obvious and the image $\Ac(\Mb;N)$ of $\alpha_{\Mb;\Nb}$ may be identified as the subalgebra of $\Ac(\Mb)$ generated by
$\Phi_\Mb(f)$ for $f\in\CoinX{N}$.  See~\cite{BrFrVe03,FewVerch_aqftincst:2015} for details on these and other properties. If $N$ contains a Cauchy surface of $\Mb$ then by choosing $\chi\in C^\infty(M)$ appropriately we may write any $\Phi_\Mb(f)=\Phi_\Mb(P\chi E_P f)$ with $P\chi E_P f\in \CoinX{N}$, thus establishing the time-slice property. 
  
Einstein causality holds owing to the support properties of the Green functions, 
while the Haag property is proved in Appendix~\ref{appx:Haagduality} for reference, as we do not know of a proof in the purely $*$-algebraic framework (see, e.g.,~\cite{Camassa:2007} and references therein for the von Neumann algebraic approach to Haag duality for spacetimes that are parts of Minkowski spacetime).

The uncoupled probe system is defined in exactly the same way, replacing $P$ and $E_P$ by $Q$ and $E_Q$, and gives an algebra $\Bc(\Mb)$, whose generators will be denoted $\Psi(f)$ for $f\in\CoinX{N}$. As in our general discussion, the system and probe may be treated as a single uncoupled system with algebra $\Ac(\Mb)\otimes\Bc(\Mb)$, where $\otimes$ denotes the algebraic tensor product.
(This algebra was denoted $\Uc(\Mb)$ in the Introduction.) Equivalently, this may be regarded as the quantization of the Green-hyperbolic operator $P\oplus Q$, with generators $\Xi_0(F)$ labelled by $F\in \CoinX{M;\CC^2}\cong \CoinX{M}\oplus\CoinX{M}$. The correspondence with the tensor product form is given by
\begin{equation}
\Xi_0(f\oplus h)= \Phi(f)\otimes \II_{\Bc(\Mb)} + \II_{\Ac(\Mb)}\otimes \Psi(h).
\end{equation}

Finally, the quantized coupled field--probe system has an algebra $\Cc(\Mb)$ obtained 
in the same way, using $T$ and the corresponding advanced-minus-retarded solution operator $E_T$, and writing the generators as $\Xi(F)$ 
for $F\in \CoinX{M;\CC^2}$. As the operator $T$ agrees with $P\oplus Q$ outside $K$, the restriction of $E_T$ to any causally convex region $L$ contained in $K^\perp$ agrees with the restriction of $E_{P\oplus Q}$. In consequence, the algebras
$\Ac(\Lb)\otimes\Bc(\Lb)$ and $\Cc(\Lb)$ are isomorphic under the map defined by 
\begin{equation}
\chi_\Lb(\Xi_{0,\Lb}(F)) = \Xi_\Lb(F), \qquad (F\in\CoinX{L})
\end{equation}
where we again decorate fields with subscripts to indicate the relevant spacetime. 
This definition is clearly compatible with the requirements summarised in the diagram~\eqref{eq:naturality}. 

\paragraph{Scattering morphisms} The scattering morphism $\Theta$ of our model is obtained by comparing the quantized dynamics of $T$ with that of $P\oplus Q$. By the time-slice property, it is enough to specify the action of $\Theta$ on a generator $\Xi_0(F)$ of $\Ac(\Mb)\otimes\Bc(\Mb)$ where $F$ is compactly supported in $M^+$. A standard argument (see e.g.,~\cite{FewVerch_aqftincst:2015}) can be adapted to give the result in the following form: 
 \begin{equation}\label{eq:rcegen_text}
 \Theta\Xi_0(F) = \Xi_0(F - \tilde{R} E_T F) = \Xi_0(F - \tilde{R} E_T^-F )
 \end{equation}
 for all $F\in\CoinX{M^+;\CC^2}$, where
 \begin{equation}
 \tilde{R} = T-P\oplus Q = \begin{pmatrix}
 0 & R \\ R & 0
 \end{pmatrix}.
 \end{equation}
 For completeness,~\eqref{eq:rcegen_text} is proved in Appendix~\ref{appx:rce}, where it is also established that the causal factorisation formula~\eqref{eq:Bogoliubov} holds when
 considering two probe systems with causally orderable coupling regions.

\section{Application of the measurement scheme to the detector model}\label{sec:applic}

\subsection{Induced observables}

We are now in a position to apply the general analysis of measurement schemes
to our field--probe system. In particular, we will compute the induced system
observables $\varepsilon_\sigma(B)$ for various choices of probe observable $B\in \Bc(\Mb)$ and also show how the product on $\Bc(\Mb)$ can be deformed to make
$\varepsilon_\sigma$ a (noninjective) $*$-homomorphism. Some aspects of the general analysis cannot be illustrated with our example, as quantised above, because the algebras do not contain nontrivial effects. This can be addressed by following a $C^*$-algebraic quantisation and is left for a separate study elsewhere.

The first step is to compute $\Theta (\II\otimes \Psi(h))$ for
a general $h\in\CoinX{M}$, which may be assumed for convenience to be supported in $M^+$ (failing which, we use the time-slice property to replace $h$ by a test function supported in $M^+$ that gives the same smeared field). Then~\eqref{eq:rcegen_text} gives 
\begin{equation}\label{eq:Theta_act}
\Theta(\II\otimes \Psi(h)) = \Phi(f^-)\otimes\II + \II\otimes \Psi(h^-),
\end{equation}
where
\begin{equation}\label{eq:fin}
\begin{pmatrix}  f^- \\ h^- \end{pmatrix}
=\begin{pmatrix} 0 \\ h \end{pmatrix}
-\begin{pmatrix} 0 & R \\ R & 0 \end{pmatrix}   E_T^-
\begin{pmatrix} 0 \\ h \end{pmatrix} \qquad (h\in\CoinX{M^+}).
\end{equation}
Clearly, $f^-$ is supported within $\supp \rho$, while $h^-$ is compactly supported within $M^+\cup\supp\rho$. 
Note also that replacing $h$ by $h+Qh'$, with $h'\in\CoinX{M^+}$, leaves $f^-$ unchanged and modifies $h^-$ to $h^-+Qh'$,
leaving $\Psi(h^-)$ unchanged. In this way one sees that~\eqref{eq:Theta_act} depends on $h$ only via $\Psi(h)$, provided $h\in\CoinX{M^+}$.
We immediately obtain $\Theta(\II\otimes \Psi(h)^n)$ because $\Theta$ is a homomorphism. These results may be summarised by the identities
\begin{equation}
\Theta (\II\otimes e^{i\Psi(h)}) = e^{i\Phi(f^-)}\otimes e^{i\Psi(h^-)}
\end{equation}
between formal power series in $h\in\CoinX{M^+}$. The induced observables are now easily computed: if $\sigma$ is any probe state, then
\begin{equation}\label{eq:linear_equiv}
\varepsilon_\sigma(\Psi(h)) =
 \eta_\sigma(\Theta(\II\otimes  \Psi(h)))
= \Phi(f^-) + \sigma(\Psi(h^-))\II
\end{equation}
and more generally we may compute all $\varepsilon_\sigma(\Psi(h)^n)$ from the 
formal power series expression
\begin{equation}\label{eq:genfn_equiv}
\varepsilon_\sigma(e^{i\Psi(h)}) = \eta_\sigma(\Theta(\II\otimes e^{i\Psi(h)}))
=
\sigma(e^{i\Psi(h^-)}) e^{i\Phi(f^-)}.
\end{equation}
For example,~\eqref{eq:linear_equiv} is the first order term in this expansion, while the second order in $h$ gives
\begin{equation}\label{eq:quadratic_equiv}
\varepsilon_\sigma(\Psi(h)^2) = \Phi(f^-)^2 + \sigma(\Psi(h^-)^2)\II +2\sigma(\Psi(h^-))\Phi(f^-),
\end{equation}
which may be confirmed by direct calculation.
Other induced observables can be obtained by taking suitable functional derivatives of the above expressions.

We can use the results just obtained to quantify the additional unsharpness introduced by the measurement scheme. The variance of the measured observable
$\widetilde{\Psi(h)}$ in the state $\utilde{\omega}_\sigma$ is
\begin{align}
\Var(\widetilde{\Psi(h)};\utilde{\omega}_\sigma) &= \utilde{\omega}_\sigma\left(\widetilde{\Psi(h)}^2\right) - 
\utilde{\omega}_\sigma\left(\widetilde{\Psi(h)}\right)^2 \notag\\
&=
\omega(\varepsilon_\sigma(\Psi(h)^2))- \omega(\varepsilon_\sigma(\Psi(h)))^2 \notag\\
&=
\omega(\Phi(f^-)^2)+ 2\omega(\Phi(f^-))\sigma(\Psi(h^-))+ \sigma(\Psi(h^-)^2)
\notag \\
&\qquad\qquad -\left(\omega(\Phi(f^-))^2 
+ 2\omega(\Phi(f^-))\sigma(\Psi(h^-)) + \sigma(\Psi(h^-))^2\right)\notag\\ 
&= \Var(\Phi(f^-);\omega)+\Var(\Psi(h^-);\sigma),
\end{align}
which clearly shows that the additional variance can be attributed to fluctuations
in the probe. Indeed, because $\tau^+$ in~\eqref{eq:Btilde} is a homomorphism, we deduce that
\begin{equation}
	\utilde{\omega}_\sigma(e^{i\widetilde{\Psi(h)}}) = 	\utilde{\omega}_\sigma(\widetilde{e^{i\Psi(h)}}) = \omega(\varepsilon_\sigma(e^{i\Psi(h)})) = \sigma(e^{i\Psi(h^-)}) \omega(e^{i\Phi(f^-)}).
\end{equation}
Therefore, if $\omega$ and $\sigma$ are sufficiently regular that $\lambda\mapsto 
	\omega(e^{i\lambda\Phi(f^-)})$ and 
	$\lambda\mapsto \sigma(e^{i\lambda\Psi(h^-)})$ are characteristic 
	functions of probability measures for measurements of $\Phi(f^-)$ in state $\omega$ and $\Psi(h^-)$ in state $\sigma$, then the probability measure for
	measurement outcomes of $\widetilde{\Psi(h)}$ in state $\utilde{\omega}_\sigma$ 
	has characteristic function 
	$\lambda\mapsto \utilde{\omega}_\sigma(e^{i\lambda\widetilde{\Psi(h)}})$ and is therefore
	the convolution of these measures. This quantifies precisely the way in which the actual measurement is less sharp than the induced observable, and clearly attributes its origin to fluctuations in the probe.

An important class of special cases arises when $\sigma$ is a quasifree state with (possibly nonvanishing) one-point function and truncated two-point function
\begin{align}\label{eq:qfVS}
V(h_1)&=\sigma(\Psi(h_1)) \\
S(h_1,h_2)&= \sigma(\Psi(h_1)\Psi(h_2))-V(h_1)V(h_2).
\end{align} 
In this case,~\eqref{eq:genfn_equiv} becomes
\begin{equation}\label{eq:genfunQF}
\varepsilon_\sigma(e^{i\Psi(h)}) =   e^{iV(h^-)-S(h^-,h^-)/2} e^{i\Phi(f^-)},
\end{equation}
while~\eqref{eq:linear_equiv} and~\eqref{eq:quadratic_equiv} simplify to 
\begin{align}
\varepsilon_\sigma(\Psi(h)) &=  \Phi(f^-) + V(h^-)\II \\
\varepsilon_\sigma(\Psi(h)^2) &= \Phi(f^-)^2 +2V(h^-)\Phi(f^-)+ \left(S(h^-,h^-)+V(h^-)^2\right)\II. 
\end{align}

As described in Section~\ref{sec:genscheme}, $\varepsilon_\sigma$ 
maps probe observables to induced observables of the system. 
The map is not injective: if $h$ is supported outside $J^+(K)$,
$E_T^- (0\oplus h)$ vanishes on $\supp\rho\subset K$. Thus $f^-=0$ and the formal series
$\varepsilon_\sigma(e^{i\Psi(h)})$ is a multiple of the unit; 
therefore the same is true of $\varepsilon_\sigma(\Psi(h)^n)$ for each $n$.
The correspondence established by $\varepsilon_\sigma$ is not an algebra homomorphism either. However, by deforming the product on $\Bc(\Mb)$ one can make $\varepsilon_\sigma$ into a homomorphism, which is convenient for understanding
the nature of the correspondence. Again, it is useful to use generating functions to express the modified product, which will be denoted by a $\star$. 
As formal series in $h,h'\in\CoinX{M^+}$, we compute
\begin{align}
\varepsilon_\sigma(e^{i\Psi(h)})\varepsilon_\sigma(e^{i\Psi(h')})  
&=\sigma(e^{i\Psi(h^-)})\sigma(e^{i\Psi(h^{\prime -})}) e^{i\Phi(f^-)} e^{i\Phi(f^{\prime -})} \notag\\
&= \sigma(e^{i\Psi(h^-)})\sigma(e^{i\Psi(h^{\prime -})})  e^{-iE_P(f^-,f^{\prime -})/2}e^{i\Phi(f^-+f^{\prime -})} \notag\\
&=G(h^-+h^{\prime -}) G(h^-)^{-1}G(h^{\prime -})^{-1}
e^{-iE_P(f^-,f^{\prime -})/2} \varepsilon_\sigma(e^{i\Psi(h+h')})\label{eq:prestar}
\end{align}
where $G(h)=\sigma(e^{i\Psi(h)})^{-1}$ with the inverse computed as a formal power series. This
inverse exists and is uniquely determined because $\sigma(e^{i\Psi(h)})=1+O(h)$;\footnote{This marks a point at which it is important that we study formal series. There are states on the Weyl algebra (see e.g., the tracial state described in~\cite[\S 2.1]{Fewster_artofthestate:2018}) for which some or all Weyl generators other than the unit have vanishing expectation value.} furthermore, one has the identities
\begin{equation}
G(\lambda (h+Qh'))=G(\lambda h), \qquad G(\lambda\bar{h})= \overline{G(-\lambda h)}
\end{equation}
of formal series in $\lambda\in\RR$ for all $h,h'\in\CoinX{M}$. 
By the comments following~\eqref{eq:fin} it is clear that the right-hand side of~\eqref{eq:prestar}
depends on $h,h'$ only via $\Psi(h)$ and $\Psi(h')$, and so we may define a deformed product on $\Bf(\Mb^+)=\Bf(\Mb)$ by the identity 
\begin{align}\label{eq:star}
e^{i\Psi(h)} \star e^{i\Psi(h')} &=   G(h^-+h^{\prime -}) G(h^-)^{-1}G(h^{\prime -})^{-1} e^{-i E_P(f^-,f^{\prime -})/2}  e^{i\Psi(h+h')}
\end{align}
of formal series in $h,h'\in\CoinX{M^+}$. One easily checks that the $\star$-product is compatible with the $*$-operation, in the sense that
\begin{equation}
\left( e^{i\Psi(h)} \star e^{i\Psi(h')}\right)^* = \left( e^{i\Psi(h')}\right)^* \star \left(e^{i\Psi(h)}\right)^*.
\end{equation}
The deformed star product is, in fact, gauge-equivalent (in the sense of~\cite{Kontsevich:2003}) to the product on the algebra obtained by quantising the solution space $\CoinX{M}/Q\CoinX{M}$ with respect to the pre-symplectic form  
\begin{equation}
\nu([h],[h']) := E_P(f^-,f^{\prime -}),\qquad h,h'\in\CoinX{M^+}.
\end{equation}
As a particular example, we may consider the situation where $\sigma$ is a quasi-free state  
described by~\eqref{eq:qfVS}. Then the formal series $G(h)$ converges, $G(h)= \exp(-iV(h)+S(h,h)/2)$, whereupon the $\star$-product is defined by
\begin{align}
e^{i\Psi(h)} \star e^{i\Psi(h')} &=   e^{S_{\sym}(h^-,h^{\prime -})-i E_P(f^-,f^{\prime -})/2}  e^{i\Psi(h+h')}
\end{align}
for $h,h'\in\CoinX{M^+}$, where $S_\sym=\frac{1}{2}(S + S^t)$ is the symmetric part of the probe's truncated two-point function. 

Returning to the general case, the existence of the $\star$-product implies at once that the subspace of induced system observables $\varepsilon_\sigma(\Bc(\Mb))$ is in fact a $*$-subalgebra 
of $\Af(\Mb)$, because it is the image of a $*$-homomorphism. 
An important point is that causal disjointness does not imply $\star$-commutativity. Indeed, one has
\begin{equation}
[\Psi(h),\Psi(h')]_\star = i E_P(f^-,f^{\prime -})\II,
\end{equation}
which reflects the fact that induced system observables of causally disjoint probe measurements do not necessarily commute. Physically, the interaction can create correlations between degrees of
freedom of the probe that are then observed at spacelike separation.  We remark that the occurrence of `deformed' operator products in typical measurement interactions has been observed elsewhere \cite{Andersson:2013},  cf.\ also \cite{BuschLahti-StandModel}.

\subsection{Localisation of induced observables}

On general grounds (Theorem~\ref{thm:localisation}), the induced observable $\varepsilon_\sigma(\Psi(h))$ may be localised in any connected open causally convex neighbourhood of the causal hull of the interaction region. 
This may be seen explicitly in our concrete example, using~\eqref{eq:fin} to write $f^-$ as a product
\begin{equation}
f^- = -\rho (E_T^- (0\oplus h))_2,
\end{equation}
where the subscript denotes the second component and we assume $h\in\CoinX{M^+}$. Any localisation region for $\Phi(f^-)$ is a localisation region for $\varepsilon_\sigma(\Psi(h)^n)$ ($n\in\NN_0$). In particular, this applies to any open causally convex region containing $\supp \rho\cap J^-(\supp h)$. More generally, every induced system observable can be localised in any causally convex neighbourhood of $\supp \rho$; informally we might regard the causal hull of the coupling region as providing a common minimal localisation region.\footnote{Using the
time-slice property these observables can of course be localised in any causally convex neighbourhood of any Cauchy surface. The important point here is that
they {\em can} be localised in neighbourhoods of the coupling region.} 

One might wonder whether a tighter localisation
is possible. For instance, one might be tempted to say that $\varepsilon_\sigma(\Psi(h))=\Phi(f^-)$ may be localised in $\supp f^-$, 
even if the support is not causally convex. As we now argue, however, certain properties of $\varepsilon_\sigma(\Psi(h))$ are sensitive 
to the geometry of the whole causal hull of $\supp f^-$, so it does not seem useful to assert any tighter localisation.
 
A fundamental question is whether two induced observables, e.g., $\varepsilon_\sigma(\Psi(h))$ and $\varepsilon_\sigma(\Psi(h'))$ are compatible or not. Computing their commutator, we find
\begin{align}\label{eq:comm}
[\varepsilon_\sigma(\Psi(h)),\varepsilon_\sigma(\Psi(h'))] &= 
[\Phi( f^-),\Phi( f^{\prime -})] 
= 
iE_P( f^-,  f^{\prime -})\II .
\end{align}
Now the right-hand side of this equation is sensitive to changes in the geometry in
\begin{equation}
S = (J^+(\supp f^-)\cap J^-(\supp f^{\prime -} ))\cup (J^-(\supp f^-)\cap J^+(\supp f^{\prime -} )).
\end{equation}
Fixing $h$, there will (in generic cases) be other test functions $h'\in\CoinX{M^+}$ so that $f^{\prime -}$ and $f^-$ have the same support, in which case $S$ coincides with the causal hull of $\supp  f^-$
and may contain points that are outside $\supp  f^-$ or even $\supp\rho$. Given that the 
question of compatibility of the two induced observables can be sensitive to the geometry outside $\supp f^-$ (unless it is already causally convex) it would seem inappropriate to declare that they are, nonetheless, local to that region.
We see that a naive `off-shell' localisation provided by the support of the test function can be misleading;
it would also be incompatible with local covariance~\cite{BrFrVe03,FewVerch_aqftincst:2015}.
By contrast, the right-hand side of \eqref{eq:comm} is insensitive to changes in the geometry outside the causal hull of $\supp\rho$, for all $h,h'$, and so the causal hull is a 
legitimate localisation region for arbitrary induced observables.

For illustrative purposes, if the coupling were singularly supported along a timelike curve segment $\gamma:[0,\tau]\to M$, then the minimal localisation region 
would be $J^+(\gamma(0))\cap J^-(\gamma(\tau))$. In the specific example of an eternally uniformly accelerated
probe in Minkowski spacetime, as in the traditional Unruh--deWitt analysis, the localisation region is an entire wedge region (and not, for example, the curve itself). 

Summarising, while the induced observables are all localised in (neighbourhoods of) the causal hull of the coupling region, there is in general no concept of their localisation that is simultaneously tighter and useful.

\subsection{Perturbative treatment of the detector response}

It is natural to consider measurements in which the coupling between the system and probe is kept weak to minimise disruption to the system, and the weak-coupling regime has often been studied in 
standard treatments of the detector system~\cite{Unruh:1976,DeWitt:1979}. To conclude, we therefore consider this regime in our model. 

First note that the Green operator $E_T^-$ obeys the equation
\begin{align}
E_T^-F &= E_{P\oplus Q}^-F - E_{P\oplus Q}^-\tilde{R}E_T^- F \notag\\
&= E_{P\oplus Q}^-F - E_{P\oplus Q}^-\tilde{R}E_{P\oplus Q}^- F + 
E_{P\oplus Q}^-\tilde{R}E_{P\oplus Q}^-\tilde{R}E_T^- F ,
\end{align}
where the second line arises by substituting the first back into itself and
the first on rewriting the equation $T\Xi=F$ as $(P\oplus Q)\Xi=F-\tilde{R}\Xi$ and using uniqueness of solution with future compact support~\cite{Baer:2015} (see Appendix~\ref{appx:rce} for the definition of `future compact').  Now replacing $\rho$ by $\lambda\rho$, with $|\lambda|\ll 1$, we obtain a Born expansion
\begin{equation}\label{eq:Born}
E_{T}^-F = E_{P\oplus Q}^-F - \lambda E_{P\oplus Q}^-\tilde{R}E_{P\oplus Q}^- F + 
\lambda^2 E_{P\oplus Q}^-\tilde{R}E_{P\oplus Q}^-\tilde{R}E_{P\oplus Q}^- F + O(\lambda^3), 
\end{equation}
where $\tilde{R}$ is given in terms of $\rho$ using the formulae above. 
In fact, because
\begin{equation}
\tilde{R}E_{P\oplus Q}^- = \begin{pmatrix}
0 & R E_Q^- \\ R E_P^- & 0
\end{pmatrix}
\end{equation}
is off-diagonal,~\eqref{eq:fin} reduces to
\begin{equation}
\begin{pmatrix}
f^- \\ h^-
\end{pmatrix} = \begin{pmatrix}
0 \\ h \end{pmatrix} - \lambda \tilde{R}
E_{T}^-\begin{pmatrix}
0 \\ h \end{pmatrix} =
\begin{pmatrix}
-\lambda \rho E_Q^- h + O(\lambda^3)\\
h + \lambda^2 \rho E_P^-\rho E_Q^- h +O(\lambda^4)
 \end{pmatrix} ,
\end{equation}
for $h\in\CoinX{M^+}$, so
\begin{align}
\Phi(f^-) & = -\lambda \Phi(\rho E_Q^- h) + O(\lambda^3) \\
\Psi(h^-) & = \Psi(h + \lambda^2 \rho E_P^-\rho E_Q^- h) + O(\lambda^4)  .
\end{align}

Therefore, assuming that $\sigma$ has vanishing-one-point function and $h\in\CoinX{M^+}$,
\begin{equation}
\varepsilon_\sigma(\Psi(h)^*\Psi(h)) = \Phi(f^-)^* \Phi(f^-)
+\sigma(\Psi(h^-)^*\Psi(h^-))\II
\end{equation}
has an expectation value
\begin{equation}\label{eq:Born_result}
\omega(\varepsilon_\sigma(\Psi(h)^*\Psi(h))) = S(\overline{h},h)  +  \lambda^2 \left(
W(\overline{h_1},h_1) + 2\Re S(\overline{h},h_2) \right) + O(\lambda^4),
\end{equation}
where $S$ and $W$ are the two-point functions of $\sigma$ and $\omega$, and 
\begin{equation}
h_1 = \rho E_Q^- h, \qquad
h_2 = \rho E_P^- \rho E_Q^- h 
\end{equation}
are compactly supported within $\supp\rho$. In this expression, the 
lowest order term describes the spontaneous excitation of the probe in the absence of any coupling and can be regarded as background noise. 

Using this result, we can make some contact with the traditional analysis of the
Unruh--deWitt detector. (See also Smith's treatment in Sec.~3.3 of his thesis~\cite{Smith_thesis}.) Suppose that the GNS representation of the probe induced by $\sigma$ is a Fock representation and that $h$ is chosen so that $\Psi(h)$ closely approximates an annihilation operator in this representation. For purposes of exposition, let us suppose it actually \emph{is} an annihilation operator, whereupon
the two terms in~\eqref{eq:Born_result} involving $S$ both vanish, and the left-hand side is actually the expectation of a number operator $N_h$ for the mode annihilated by $\Psi(h)$ (up to normalisation). This gives
\begin{equation}
\omega(\varepsilon_\sigma(N_h)) = \lambda^2 W(\overline{h_1},h_1) + O(\lambda^4)
\end{equation}
Now let $\rho$ become concentrated along a timelike worldline $\gamma(\tau)$, in a proper time parametrization, so that $\rho$ becomes a compactly supported distribution acting by
\begin{equation}
\rho (f) = \int_\RR f(\gamma(\tau)) \tilde{\rho}(\tau)\, d\tau ,\qquad f\in C^\infty(M)
\end{equation}
for some smooth $\tilde{\rho}$, which we take to be real-valued. In this limit, one obtains
\begin{equation}
\omega(\varepsilon_\sigma(N_h)) = \lambda^2 
 \left((\gamma\times\gamma)^*W\right)(\tilde{\rho} \gamma^* E_Q^-\overline{h},\tilde{\rho}\gamma^*E_Q^-h) + O(\lambda^4).
\end{equation}  
Supposing further that $\gamma^* E_Q^-h$ agrees with $e^{iE\tau}$ on $\supp\tilde{\rho}$, we have 
the final answer
\begin{equation}
\omega(\varepsilon_\sigma(N_h)) = \lambda^2 \int d\tau\,d\tau'\, e^{-iE(\tau-\tau')}\tilde{\rho}(\tau)\tilde{\rho}(\tau') W(\gamma(\tau),\gamma(\tau')) + O(\lambda^4).
\end{equation}  
The various approximations and limits employed here are made to approximate the interaction of the field with a two-level quantum mechanical system with an energy gap $E$ (represented by a specific mode of the probe field in our model). A more direct treatment of the Unruh--deWitt system will be given elsewhere. Nonetheless, we may observe that 
the coefficient of $\lambda^2$ the response function of a switched
Unruh--deWitt detector, as computed using standard first order perturbation theory (see, e.g.,~\cite[\S 3]{FewsterJuarezAubryLouko:2016}). 
The calculation above can and will be discussed in more detail elsewhere,
in the context of an appraisal of the Unruh effect in the light of our formalism,
and also the previous results of 
\cite{Unruh:1976,GroveOttewill:1983,UnruhWald-what,deBievreMerkli:2006,CrispinoHiguchiMatsas:2008,FewsterJuarezAubryLouko:2016} among others. However we can already see that our formalism
provides a more local viewpoint on the detector. The title of~\cite{FewsterJuarezAubryLouko:2016} was `Waiting for Unruh': the results above indicate \emph{where} one should wait, namely in the 
intersection of $M^+$ and $J^+(K)$.

\section{Conclusions}\label{sect:conc}

We have proposed a general framework for quantum measurement within the local covariant setting of quantum field theory. Thereby,
we have reconciled general relativity and quantum field theory with quantum measurement theory --- at least partly, since spacetime and its causality structure 
enters as a given background, so that the process of measurement doesn't have any dynamical influence on the spacetime. The central 
element of our approach is the localized dynamical coupling of the `system' and the `probe'. As indicated, under very general and 
natural assumptions, this coupling gives rise to a scattering morphism, so that the measurement interaction by the probe on the 
system can be seen as subjecting the system to a scattering process, depending on the probe observable measured and the initial
state in which the probe has been prepared prior to measurement. This is very much in the spirit of the operational approach to 
quantum measurement~\cite{Busch_etal:quantum_measurement}
and realizes a measurement scheme;
however we have focussed on the induced observables of the system for given coupling, probe observables and initial probe states, 
rather than the converse problem of finding a coupled system and probe observable corresponding to a given system observable (cf.~\cite{Ozawa:1984,OkamuraOzawa}).
The converse problem is left open for further study, but in our view it is of particular interest to understand what can be measured using physically realisable models.
We have shown that 
the localization properties of the induced observables in our scheme are determined by the coupling regions between system and probe: the induced observable is 
localized in the causal hull of the coupling region. 

In a further step, we have formulated the state change of the system consequent upon a probe measurement in terms of instruments depending on the measured probe observable
and initial probe state, which can be used to describe both selective or non-selective measurements. We have argued that, in general, 
the post-selected state differs everywhere in spacetime from the original system state, in view of the genericity of system states that have the Reeh--Schlieder property.
By contrast,
the updated state following non-selective measurement agrees with the original state in the causal complement of the coupling region. 
Nevertheless, we have also shown that the causally ordered composition of a pair of instruments coincides with the instrument of the composed measurement, in situations where the second coupling region does not precede the first, and assuming a 
causally factorizing scattering morphism. The 
order of composition of instruments is therefore irrelevant if the coupling regions are causally disjoint. This result,
Thm.~\ref{thm:causalcomp}, shows the consistency of the framework with the principles of locality and measurement. It also renders moot the discussion where and when 
a state change of the system takes place as consequence of a measurement, and we see it as the core result of the present article.

Furthermore, in the example consisting of two linear quantized fields brought to interaction by a localized coupling function, the induced system observables
can be described precisely for simple probe observables. 
The results agree, in suitable limits, with the perturbative treatment of a switched Unruh-deWitt detector (see, e.g.,~\cite{FewsterJuarezAubryLouko:2016}). 
The localisation of induced observables has been given in terms of causally convex subsets within the causal hull of the coupling region, and we have pointed out that there is no
viable sharper concept of localization. The causal hull of the coupling region is therefore the 
minimal common localization region for all the induced system observables.

The framework discussed in this article promises to have further applications and to admit extensions. It should, in particular, shed light on the localization 
properties of induced observables of measurements conducted with Unruh--deWitt detectors in arbitrary motion, and lead to a structural understanding of the relation
of the coupling and the spacetime structure to the state of the detector (the probe) after measurement. For example, one can speculate whether 
the eventual thermality of the probe is linked to the existence of a non-void causal complement for its trajectory.  Our framework has already been applied to the issue of measurement and causality in a discussion of so-called `impossible measurements'~\cite{BostelmannFewsterRuep:2020}.
It may be that it can also shed light on discussions of quantum theory where some causality violation is permitted, as has been discussed in quantum information contexts~\cite{TolksdorfVerch:2018}. In an extended framework, the influence of the measurement process on the spacetime structure ought to be taken into account. 
Furthermore, one may expect that extensions of the framework might ultimately provide hints towards what is needed to discuss measurement 
in the context of quantum gravitational theories, where not only the system but also the spacetime and its causal structure have to
be inferred from the results of measurements. 

\bigskip\bigskip
\begin{footnotesize}
	\noindent{\bf Acknowledgements}
	We both thank the Mathematisches Forschungsinstitut Oberwolfach and the organisers of the workshop `Recent Mathematical Developments in Quantum Field Theory' (July 2016) 
	at which this work commenced, and the Banff International Research Station and the Max Planck Institute for Mathematics in the Sciences along with the respective organisers of the workshops
	`Physics and Mathematics of Quantum Field Theory' (July/August 2018) and `Progress and Visions in Quantum Theory in View of Gravity' (October 2018) during which it progressed. 
	CJF thanks the Institute for Theoretical Physics at the University of Leipzig and the Institute of Mathematics at the University of Potsdam for hospitality during 
	part of the course of this work, while RV reciprocally thanks the Department of Mathematics at the University of York for its hospitality and financial support. 
	We both thank Dorothea Bahns, 
	Henning Bostelmann, Detlev Buchholz, Roger Colbeck, J\"{u}rg Fr\"{o}hlich, Eli Hawkins, 
	Leon Loveridge, Albert Georg Passegger, Kasia Rejzner, Maximilian Ruep, Yoh Tanimoto and Robert Wald 
	for remarks and questions that have helped us to clarify our presentation.\par
\end{footnotesize}

\appendix
\section{General properties of the scattering morphism}\label{appx:genscattering}

We prove Proposition~\ref{prop:scattering}, which restate here in terms of morphisms. The notation established in Sec.~\ref{sec:abstract} is used freely.
\begin{proposition}
	(a) If $\hat{K}$ is any compact set containing the coupling region $K$, let $\hat{\Theta}$ be the morphism obtained if one replaces $K$ by $\hat{K}$ in the construction of the scattering morphism of a given coupled theory. Then $\hat{\Theta}=\Theta$.
	(b) If $L$ is an open causally convex subset of the causal complement $K^\perp = M^+\cap M^-$ of $K$, then $\Theta$ acts trivially on $\Uc(\Mb;L)=\Ac(\Mb;L)\otimes \Bc(\Mb;L)$, i.e.,
	\begin{equation}
	\Theta \circ (\alpha_{\Mb;\Lb}\otimes\beta_{\Mb;\Lb}) = \alpha_{\Mb;\Lb}\otimes\beta_{\Mb;\Lb}.
	\end{equation}
	(c)	Suppose that $L^+$ (resp. $L^-$) is an open causally convex subset of $M^+$ (resp., $M^-$), and that $L^+\subset D(L^-)$. Then $\Theta\circ (\alpha_{\Mb;\Lb^+}\otimes\beta_{\Mb;\Lb^+})$ factors via
	$\alpha_{\Mb;\Lb^-}\otimes\beta_{\Mb;\Lb^-}$, i.e., $\Theta \Uc(\Mb;L^+) \subset \Uc(\Mb;L^-)$. 
\end{proposition}
\begin{proof}
	(a) Note that $\hat{M}^\pm=M\setminus J^\mp(\hat{K})$ are subsets of $M^\pm$, 
	giving $(\alpha^\pm){}^{-1}\circ \hat{\alpha}^\pm = \alpha_{\Mb^\pm;\hat{\Mb}^\pm}$ by the rule for consecutive inclusions and the time-slice property. We may then calculate
	\begin{align}
	\kappa^\pm\circ(\alpha^\pm\otimes\beta^\pm)^{-1}\circ(\hat{\alpha}^\pm\otimes\hat{\beta}^\pm)
	&=\gamma^\pm\circ\chi^\pm\circ (\alpha_{\Mb^\pm;\hat{\Mb}^\pm}\otimes \beta_{\Mb^\pm;\hat{\Mb}^\pm}) = \gamma^\pm\circ \gamma_{\Mb^\pm;\hat{\Mb}^\pm}\circ\hat{\chi}^\pm = \hat{\gamma}^\pm\circ\hat{\chi}^\pm \notag\\
	&=\hat{\kappa}^\pm
	\end{align}
	using~\eqref{eq:naturality} and~\eqref{eq:functorial}. It follows that $\tau^\pm=\hat{\tau}^{\pm}$ and hence $\hat{\Theta}=\Theta$.
	
	(b,c) First suppose that $L\subset M^\pm$ is an open causally convex set.
	Then  
	\begin{equation} 
	\kappa^\pm\circ(\alpha_{\Mb^\pm;\Lb}\otimes\beta_{\Mb^\pm;\Lb}) = \gamma^\pm\circ \chi_{\Mb^\pm}\circ (\alpha_{\Mb^\pm;\Lb}\otimes\beta_{\Mb^\pm;\Lb}) = 
	\gamma^\pm \circ\gamma_{\Mb^\pm;\Lb}\circ\chi_{\Lb} =\gamma_{\Mb;\Lb}\circ\chi_{\Lb} ,
	\end{equation}
	in which we have used~\eqref{eq:naturality}. Then 
	the definition of $\tau^\pm$ gives
	\begin{align}\label{eq:useful}
	\tau^\pm\circ (\alpha_{\Mb;\Lb}\otimes\beta_{\Mb;\Lb}) &= \kappa^\pm\circ  (\alpha^\pm\otimes\beta^\pm)^{-1}\circ (\alpha_{\Mb;\Lb}\otimes\beta_{\Mb;\Lb}) 
	=\kappa^\pm\circ  (\alpha_{\Mb^\pm;\Lb}\otimes\beta_{\Mb^\pm;\Lb}) \notag \\
	&=\gamma_{\Mb;\Lb}\circ\chi_{\Lb},
	\end{align}
	which we now use to prove parts (b) and (c).
	
	For (b), we assume that $L\subset M^+\cap M^-$, so we have 
	\begin{equation}
	\tau^+\circ (\alpha_{\Mb;\Lb}\otimes\beta_{\Mb;\Lb})=\tau^-\circ (\alpha_{\Mb;\Lb}\otimes\beta_{\Mb;\Lb}),
	\end{equation}
	and hence
	\begin{equation}
	\Theta \circ (\alpha_{\Mb;\Lb}\otimes\beta_{\Mb;\Lb}) = (\alpha_{\Mb;\Lb}\otimes\beta_{\Mb;\Lb}),
	\end{equation}
	that is, $\Theta$ acts trivially on $\Ac(\Mb;L)\otimes \Bc(\Mb;L)$. 
	
	Finally, for part (c) we apply~\eqref{eq:useful} to $L^+\subset M^+$ and $L^-\subset M^-$ giving
	\begin{equation}
	\tau^+\circ (\alpha_{\Mb;\Lb^+}\otimes\beta_{\Mb;\Lb^+})= \gamma_{\Mb;\Lb^+}\circ\chi_{\Lb^+},\qquad
	\tau^-\circ (\alpha_{\Mb;\Lb^-}\otimes\beta_{\Mb;\Lb^-})= \gamma_{\Mb;\Lb^-}\circ\chi_{\Lb^-},
	\end{equation}
	the first of which asserts that $\tau^+\circ (\alpha_{\Mb;\Lb^+}\otimes\beta_{\Mb;\Lb^+})$ factors through $\gamma_{\Mb;\Lb^+}$, while the second implies that $(\tau^{-})^{-1}\circ \gamma_{\Mb;\Lb^-}$ factors through $\alpha_{\Mb;\Lb^-}\otimes\beta_{\Mb;\Lb^-}$. 
	As $\gamma_{\Mb;\Lb^+}$ factors via $\gamma_{\Mb;\Lb^-}$ due to
	the assumption $L^+\subset D(L^-)$ and the timeslice property, 
	the two observations combine to show that $\Theta\circ (\alpha_{\Mb;\Lb^+}\otimes\beta_{\Mb;\Lb^+})$ factors through $\alpha_{\Mb;\Lb^-}\otimes\beta_{\Mb;\Lb^-}$, as required.
\end{proof}

\section{Properties of the maps $\eta_\sigma$ and $\varepsilon_\sigma$}\label{appx:cpetc}

We prove the properties of $\varepsilon_\sigma$ asserted in Theorem~\ref{thm:induced}, 
which rest on the fact that $\Theta$ is a unit-preserving $*$-homomorphism and the definition of $\eta_\sigma$. No originality is claimed but we did not find arguments in the literature exactly
corresponding to our situation. See, however,~\cite[\S 10.5.13]{BaumWollen:1992} for related arguments and~\cite{Busch_etal:quantum_measurement} for discussion of complete positivity. Beginning with elementary properties, we compute
\begin{equation}
\varepsilon_\sigma(\II) = \eta_\sigma(\Theta(\II\otimes\II)) = \eta_\sigma(\II\otimes\II) = \II,
\end{equation}
and then note that $\eta_\sigma((A\otimes B)^*)=\eta_\sigma(A^*\otimes B^*)=\sigma(B^*)A^*=(\eta_\sigma(A\otimes B))^*$, from which it follows that $\varepsilon_\sigma(B^*)=\varepsilon_\sigma(B)^*$. 

Next, we show explicitly that $\eta_\sigma$ is completely positive. By definition, the positive elements of a $*$-algebra are finite convex combinations of elements of the form $A^*A$. Let $N\ge 1$ and consider an element $C\in M_N(\CC)\otimes (\Ac(\Mb)\otimes\Bc(\Mb))$ given as a finite sum
\begin{equation}
C=\sum_r M_r\otimes (A_r\otimes B_r)\,.
\end{equation}
To establish complete positivity, we will show that $X = (\id_N\otimes\eta_\sigma)(C^*C)$ is positive in $M_N(\CC)\otimes\Ac(\Mb)$ and start by computing
\begin{equation}
X = \sum_{r,s} (\id_N\otimes\eta_\sigma)(M_r^*M_s\otimes A_r^*A_s\otimes B_r^*B_s) = 
\sum_{r,s} \sigma(B_r^*B_s)  M_r^*M_s\otimes A_r^*A_s\,.
\end{equation}
As $\sigma$ is a state, $\sigma(B_r^*B_s)$ is a positive matrix, and may be decomposed as 
\begin{equation}
\sigma(B_r^*B_s) = \sum_i \overline{v_r^{(i)}}v_s^{(i)}
\end{equation}
for suitable mutually orthogonal vectors $v^{(i)}$,
whereupon we find 
\begin{equation}
X=\sum_{i} W_i^*W_i\,, \qquad\text{where}\quad
W_i = \sum_s v_s^{(i)} M_s\otimes A_s\,.
\end{equation}
This proves that $\eta_\sigma$ is completely positive. The same holds for $\varepsilon_\sigma$,
because, for any $C=\sum_r M_r\otimes B_r\in M_N(\CC)\otimes\Bc(\Mb)$, 
\begin{equation}
(\id_N\otimes\varepsilon_\sigma)(C^*C) = (\id_N\otimes\eta_\sigma)(D^*D) \ge 0,
\end{equation} 
where $D= \sum_r M_r\otimes \Theta(\II\otimes B_r)$, and using complete positivity of $\eta_\sigma$.

Finally, let $B\in\Bc(\Mb)$ and define $C=\Theta (\II\otimes B)\in\Ac(\Mb)\otimes\Bc(\Mb)$, which 
can be decomposed as
\begin{equation}
C = \sum_r A_r\otimes B_r\,. 
\end{equation}
Direct calculation gives
\begin{equation}
\varepsilon_\sigma(B^*B) - \varepsilon_\sigma(B)^*\varepsilon_\sigma(B) = 
\sum_{rs} \left(\sigma(B_r^*B_s)- \overline{\sigma(B_r)}\sigma(B_s)\right) A_r^* A_s\,,
\end{equation}
and as the factors in parentheses determine a positive matrix, due to the Cauchy--Schwarz inequality, the right-hand side is a positive operator by an analogous argument to that used above.

\section{Haag property}\label{appx:Haagduality}

We establish that the quantized free scalar field, as described in the text, has the Haag property stated for general models. We use the fact that
$\Ac(\Mb)$ may be identified, as a vector space, with the symmetric tensor vector space 
\begin{equation}
\Gamma_\odot(\Sol)=\bigoplus_{n=0}^\infty \Sol^{\odot n},
\end{equation}
where $\Sol$ is the space of smooth solutions to $P\Phi=0$ with spatially compact support on $\Mb$ and $\odot$ is the symmetric tensor product. See e.g.,~\cite{FewVer:dynloc2} for further details on this viewpoint.

 The symmetric tensor vector space has a natural number operator $N$ which multiplies by $n$ on the subspace $\Sol^{\odot n}$. Now suppose $f\in\CoinX{M}$ and consider the derivation $D_f:A\mapsto [\Phi(f),A]$ of $\Ac(\Mb)$. Our aim is to determine its kernel. Due to the commutation relation Q3, $D_f$ is a lowering operator for $N$,
 \begin{equation}
 N D_f  = D_f (N-\II)
 \end{equation}
 and so $\ker D_f$ is a direct sum of its kernels within each $\Sol^{\odot n}$. Note that $D_f$ acts on a typical element $E_P h$ of $\Sol$ ($h\in\CoinX{M}$) by 
 \begin{equation}
 D_f E_P h = [\Phi(f),\Phi(h)]=iE_P(f,h)\II = \delta_f \left(E_Ph \right)\II,
 \end{equation}
where 
 \begin{equation}
 \delta_f\phi:= i\int_M\dvol \,f\phi.
 \end{equation}
 
 Suppose therefore that $D_fA=0$ and $NA=nA$. Then $A$ may be identified with an element of $\Sol^{\odot n}$, and can be further identified with a linear map 
 $\varrho_A:(\Sol^{\otimes (n-1)})^*\to \Sol$. The image $V_A=\im\varrho_A$
 is a finite dimensional \emph{support subspace} canonically associated with $A$ and it may be shown that  
 $A\in V_A^{\odot n}$ (see~\cite{FewVer:dynloc2}, especially Appendix A). In a similar way, $D_fA\in \Sol^{\odot (n-1)}$ may be identified with $n\delta_f\circ\varrho_A\in (\Sol^{\odot (n-1)})^{**}$.
 
 The assumption that $D_f A =0$ now entails that $\delta_f$ vanishes on $V_A$,
 so $A\in (\ker\delta_f)^{\odot n}$. In general, we have 
 \begin{equation}
 \ker D_f = \Gamma_\odot (\ker\delta_f)\subset \Gamma_\odot(\Sol).
 \end{equation}
 
 Now suppose that $A\in \Ac(\Mb)$ obeys $D_f A=0$ for all $f$ supported in $K^\perp$. The previous argument shows that 
 \begin{equation}
 A\in \Gamma_\odot( \Sol_K), \qquad \Sol_K = \bigcap_{f\in\CoinX{K^\perp}} \ker \delta_f
 \end{equation}
 It is easily seen that $\Sol_K$ is precisely the space of solutions with support
 contained in $M\setminus K^\perp=J(K)$. By~\cite[Lem.~3.1(i)]{FewVer:dynloc2}, if $K\subset O$ where
 $O$ is any open causally convex subset of $M$ with at most finitely many (necessarily causally disjoint) components
 then $\Sol_K\subset E_P\CoinX{O}$. This proves that $A$ may be constructed from fields smeared with test functions supported in $O$. That is, $A\in\Ac(\Mb;O)$, so the Haag property holds.
  
 \section{Scattering morphisms for the free field model}\label{appx:rce}
 
 We derive the formula~\eqref{eq:rcegen_text} for the scattering morphism used in the text, and also establish the causal factorization property. To simplify notation, it is convenient to study the scattering morphism that describes the comparison of two quantised scalar field theories on $\Mb$, based on Green hyperbolic operators $P$ and $Q$ that agree outside a compact subset $K$. We write the theories corresponding to $P$ and $Q$ as $\Ac$ and $\Bc$ respectively with generators denoted $\Phi(f)$ and $\Psi(f)$. 
 For our application in the text, $P$ is what we have written there as $P\oplus Q$, while $Q$ is the operator $T$, so $\Ac$ corresponds to $\Ac\otimes\Bc$ and $\Bc$ to $\Cc$.  We also require
 some terminology; following, e.g.,~\cite{Baer:2015}, a closed set is called \emph{future compact}
 (resp., \emph{past compact}) if it has compact intersection with all sets of the form $J^+(p)$ (resp., $J^-(p)$) as $p$ varies in $M$, and \emph{spatially compact} if it is contained in a set of the form $J(K)$ where $K$ is compact. The spaces of smooth functions with supports that are future, past, or spatially compact are designated by subscripts $fc$, $pc$, $sc$. For example, $C^\infty_{sc,fc}(M)$ consists of smooth functions on $M$ whose supports are both spatially compact and future compact.

The regions $M^\pm=M\setminus J^\mp(K)$ are causally convex open sets and therefore we may induce Green hyperbolic operators $P^\pm$ and $Q^\pm$ on $C^\infty(M^\pm)$ that must in fact agree, $Q^\pm=P^\pm$. As $M^\pm$ also contain Cauchy surfaces, there is a chain of isomorphisms
\begin{equation}\label{eq:rce_chain}
\Ac(\Mb)\to \Ac(\Mb^+) \to \Bc(\Mb^+)\to \Bc(\Mb)\to \Bc(\Mb^-)\to \Ac(\Mb^-)\to \Ac(\Mb)
\end{equation}
in which the second and fifth arise because the algebras $\Ac(\Mb^\pm)$ and $\Bc(\Mb^\pm)$
coincide, owing to the agreement of $P$ and $Q$ on $M^\pm$, while the others are instances of the time-slice property. Composing from left to right, we obtain a single overall isomorphism $\Theta:\Ac(\Mb)\to\Ac(\Mb)$.
It is enough to describe the action of $\Theta$ on a typical generator $\Phi(f)$ of $\Ac(\Mb)$ with $\supp f\subset M^+$. Given this choice, the action of the first three isomorphisms in \eqref{eq:rce_chain} simply map $\Phi(f)$ to the corresponding generator $\Psi(f)$ of $\Bc(\Mb)$. If $E_Q f=\phi^++\phi^-$ is any partition with $\phi^{+/-}\in C^\infty_{\textit{sc},\pc/\fc}(M)$ and $\supp \phi^-\subset M^-$ then $Q\phi^-=-Q\phi^+$ has support that is spatially-, future- and past-compact and
is therefore compactly supported. Further, $E_Q^-Q\phi^-=\phi^-$ while $E_Q^+Q\phi^-=-E_Q^+Q\phi^+=-\phi^+$ by uniqueness of solutions to the inhomogeneous Green-hyperbolic equations with future-/\break past-compact support, so $E_Q f=\phi^-+\phi^+ = (E_Q^--E_Q^+)Q\phi^-$. By standard properties of Green-hyperbolic operators this implies that $f-Q\phi^-\in Q\CoinX{M}$
and the axiom Q3 gives $\Psi(f)=\Psi(Q\phi^-)$. Because $Q\phi^-$ is supported in $M^-$, the last three isomorphisms in~\eqref{eq:rce_chain} map $\Psi(Q\phi^-)$ to $\Phi(Q\phi^-)$, giving overall that
\begin{equation}
\Theta\Phi(f) = \Phi(Q \phi^-).
\end{equation}  
There is considerable freedom in the choice of $\phi^-$. For example, it may be chosen (still supported in $M^-$) so that $E_Q^-f = \phi^-+\phi^0$ for some $\phi^0\in\CoinX{M}$, whereupon $E_Q f=\phi^- +\phi^0 - E_Q^+f$. 
Then 
\begin{equation}
Q\phi^- = f - Q\phi^0 = f - (Q-P)\phi^0 - P\phi^0
\end{equation}
and as $(Q-P)E_Q^+f=0$ because $f$ is supported in $M^+$, and $(Q-P)\phi^-=0$ because $\phi^-$ is supported in $M^-$, we have
\begin{equation}
Q\phi^- = f - (Q-P)E_Qf - P\phi^0.
\end{equation}
Consequently, axiom Q3 gives
\begin{equation}\label{eq:rcegengen}
\Theta\Phi(f) = \Phi(f - (Q-P)E_Qf) = \Phi(f - (Q-P)E_Q^-f)
\end{equation}
for all $f\in\CoinX{M^+}$.  

An important property of the scattering morphism is that -- as we will now show -- it factorises when the support of $\rho$ falls into two causally orderable parts. Suppose that $\rho=\rho_1+\rho_2$ where $\rho_i\in\CoinX{K_i}$ and $K_1$ and $K_2$ are compact
sets so that $J^-(K_1)$ and $J^+(K_2)$ do not intersect. Set $K=K_1\cup K_2$. 
Then $M\setminus(J^-(K_1)\cup J^+(K_2))$ contains Cauchy surfaces of $\Mb$, relative to which $K_1$ lies in the past while $K_2$ lies in the future. If $K_1$ and $K_2$ are causally disjoint, there are also Cauchy surfaces of $\Mb$ giving the reverse ordering, of course.

\begin{proposition}\label{prop:GFfactor}
	Let $P$ be a partial differential operator and suppose that Green-hyperbolic operators $Q_i$ ($i=1,2$) agree with $P$ outside compact sets $K_i$ ($i=1,2$) with
	$J^-(K_1)\cap J^+(K_2)=\emptyset$. Let $K=K_1\cup K_2$ and suppose further that 
	\begin{equation}
	Q = Q_1 + Q_2 - P
	\end{equation}
	is also Green-hyperbolic. Then the Green functions $E_i^\pm$ of $Q_i$ and $E^\pm$ of $Q$ are related by 
	\begin{equation}
	(1-(Q-P)E^-)f = (1-(Q_1-P)E_1^-)(1-(Q_2-P)E_2^-)f  
	\end{equation}
	for all $f\in \CoinX{M\setminus J^-(K)}$, and similarly, 
	\begin{equation}
	(1-(Q-P)E^+)f = (1-(Q_2-P)E_2^+) (1-(Q_1-P)E_1^+)f  
	\end{equation}
	for all $f\in \CoinX{M\setminus J^+(K)}$.
\end{proposition}
\begin{proof}
	We prove the first of the above statements; the second is proved by analogy.
	Note first that $Q$ agrees with $Q_1$ outside $K_2$, with $Q_2$ outside $K_1$, and with $P$ outside $K$.
	Take $f\in\CoinX{M\setminus J^-(K)}$ and consider the equation $Q\phi=f$, which has a unique solution with future-compact support, namely $\phi=E^-f$. 
	But also, $Q_i\phi=f-(Q-Q_i)\phi$, so 
	\begin{equation}\label{eq:QsP0}
	E^- f = \phi = E_i^- f-E_i^-(Q-Q_i)\phi.
	\end{equation}
	In particular (taking $i=2$), the second term on the right-hand side is supported in $J^-(K_1)$, and therefore $E^-f$ and $E_2^-f$ must agree in $M\setminus J^-(K_1)$. As $Q_2$ and $P$ agree outside $K_2\subset M\setminus J^-(K_1)$, it follows that 
	\begin{equation}\label{eq:QsP1}
	(Q-Q_1)E^-f=(Q_2-P)E^-f=(Q_2-P)E_2^-f=(Q-Q_1)E_2^-f. 
	\end{equation} 
	Using this identity in \eqref{eq:QsP0} with $i=1$ gives 
	\begin{equation}
	E^-f = E_1^-f - E_1^- (Q-Q_1)E^-f = E_1^-f - E_1^- (Q_2-P)E_2^-f = E_1^-(1-(Q_2-P)E_2^-)f
	\end{equation}
	and therefore one also has
	\begin{equation}
	(Q_1-P)E_1^-(1-(Q_2-P)E_2^-)f= (Q_1-P)E^-f.
	\end{equation}  
	Putting these results together, the calculation 
	\begin{align}
	(1-(Q_1-P)E_1^-)(1-(Q_2-P)E_2^-)f &= f-(Q_1-P)E_1^-(1-(Q_2-P)E_2^-)f - (Q_2-P)E_2^-f \notag\\
	&= f - (Q_1-P)E^-f -(Q_2-P)E^-f \notag\\
	&= (1-(Q-P)E^-)f,
	\end{align}
	proves the required statement.
\end{proof}  
Combining this result with~\eqref{eq:rcegengen}, we obtain immediately: 
\begin{corollary}\label{cor:Bogoliubov}
	Under the hypotheses of Prop.~\ref{prop:GFfactor}, and assuming that $P$ is also Green-hyperbolic, the scattering morphisms $\Theta$ and $\Theta_i$ comparing the quantized $Q$-dynamics (resp., $Q_i$-dynamics) to that of $P$ are related by the causal factorisation formula 
	\begin{equation}\label{eq:Bogoliubov2}
	\Theta = \Theta_1\circ\Theta_2.
	\end{equation}
\end{corollary}
Clearly, if $K_1$ and $K_2$ are causally disjoint, then we also have $\Theta=\Theta_2\circ\Theta_1$: the scattering morphisms $\Theta_i$ commute. 

Let us now apply these general results to our probe-system model. 
The uncoupled combination has dynamics given by $P\oplus Q$, while the coupled system is described by $T$. Adapting~\eqref{eq:rcegengen}, the scattering morphism $\Theta$ acts on the fields $\Xi_0$ by
\begin{equation}\label{eq:rcegen}
\Theta\Xi_0(F) = \Xi_0(F - (T-P\oplus Q)E_T F) = \Xi_0(F - \tilde{R} E_T^- F)
\end{equation}
for all $F\in\CoinX{M^+;\CC^2}$, where
\begin{equation}
\tilde{R} = \begin{pmatrix} 0 & R \\ R & 0\end{pmatrix}.
\end{equation}
This establishes the formula~\eqref{eq:rcegen_text}.
Now suppose that there are now two probe fields $\Psi_1$ and $\Psi_2$, coupled to $\Phi$ (but not each other) with compactly supported functions $\rho_i$ in regions $K_i$ with $J^-(K_1)\cap J^+(K_2)=\emptyset$. Writing the corresponding free field equations by $Q_1$ and $Q_2$, the two probes together are described by $Q=Q_1\oplus Q_2$, which is also Green-hyperbolic~\cite{Baer:2015}, while the coupling $R$ is now $RF=\rho_1 F_1 + \rho_2 F_2$. 
The causal factorisation formula~\eqref{eq:Bogoliubov2} gives
\begin{equation}
\hat{\Theta}= \hat{\Theta}_1 \circ\hat{\Theta}_2,
\end{equation}
where $\hat{\Theta}_i$ is the scattering morphism for the dynamics given by
$\hat{T}_i=P\oplus (Q_1\oplus Q_2) + \tilde{R}_i$. Identifying the overall quantized theory with $\Ac(\Mb)\otimes\Bc_1(\Mb)\otimes\Bc_2(\Mb)$,  we have 
\begin{equation}
\hat{\Theta}_1 = \Theta_{1}\otimes_3\id_{\Bc_2(\Mb)} ,  \qquad
\hat{\Theta}_2 = \Theta_{2}\otimes_2\id_{\Bc_1(\Mb)},
\end{equation}
where the $\Theta_i$ are the scattering morphisms for the quantized dynamics of
$T_i = P\oplus Q_i + \tilde{R}_i$ relative to $P\oplus Q_i$. 
This establishes the causal factorisation formula~\eqref{eq:Bogoliubov} for system-probe models of this type.  
 

{\small
\begin{thebibliography}{10}\setlength{\itemsep}{-1.5mm}
	\providecommand{\url}[1]{{#1}}
	\providecommand{\urlprefix}{URL }
	\providecommand{\href}[2]{#2}
	\expandafter\ifx\csname urlstyle\endcsname\relax
	\providecommand{\doi}[1]{DOI~\discretionary{}{}{}#1}\else
	\providecommand{\doi}{DOI~\discretionary{}{}{}\begingroup
		\urlstyle{rm}\Url}\fi
	
	\bibitem{Andersson:2013}
	Andersson, A.: {Operator Deformations in Quantum Measurement Theory}.
	\newblock Lett. Math. Phys. \textbf{104}, 415--430 (2014)
	
	\bibitem{Baer:2015}
	B{\"a}r, C.: Green-hyperbolic operators on globally hyperbolic spacetimes.
	\newblock \href{http://dx.doi.org/10.1007/s00220-014-2097-7}{Comm. Math. Phys.
		\textbf{333}(3), 1585--1615 (2015)}
	
	\bibitem{BarGinouxPfaffle}
	B{\"{a}}r, C., Ginoux, N., Pf{\"{a}}ffle, F.: Wave equations on Lorentzian
	manifolds and quantization.
	\newblock European Mathematical Society (EMS), Z{\"{u}}rich (2007)
	
	\bibitem{BaumWollen:1992}
	Baumg{\"a}rtel, H., Wollenberg, M.: Causal nets of operator algebras.
	\newblock Akademie-Verlag, Berlin (1992)
	
	\bibitem{BeniniDappiaggiSchenkel:2013}
	Benini, M., Dappiaggi, C., Schenkel, A.: Quantized {A}belian principal
	connections on {L}orentzian manifolds.
	\newblock \href{http://dx.doi.org/10.1007/s00220-014-1917-0}{Comm. Math. Phys.
		\textbf{330}(1), 123--152 (2014)}
	
	\bibitem{Bernal:2006xf}
	Bernal, A.N., S{\'{a}}nchez, M.: {Globally hyperbolic spacetimes can be defined
		as causal instead of strongly causal}.
	\newblock \href{http://dx.doi.org/10.1088/0264-9381/24/3/N01}{Class. Quantum
		Grav. \textbf{24}, 745--750 (2007)}, gr-qc/0611138
	
	\bibitem{BogoliubovShirkov}
	Bogoliubov, N., Shirkov, D.: Introduction to the Theory of Quantized Fields
	(3rd edition).
	\newblock Wiley, New York (1980)
	
	\bibitem{BostelmannFewsterRuep:2020}
	Bostelmann, H., Fewster, C.J., Ruep, M.H.: Impossible measurements require
	impossible apparatus (2020).
	\newblock ArXiv:2003.04660
	
	\bibitem{BratRob}
	Bratteli, O., Robinson, D.W.: Operator Algebras and Quantum Statistical
	Mechanics: 1, 2nd edn.
	\newblock Springer Verlag, Berlin (1987)
	
	\bibitem{BratteliRobinson}
	Bratteli, O., Robinson, D.W.: Operator Algebras and Quantum Statistical
	Mechanics: 2, 2nd edn.
	\newblock Springer Verlag, Berlin (1996)
	
	\bibitem{AdvAQFT}
	Brunetti, R., Dappiaggi, C., Fredenhagen, K., Yngvason, J. (eds.):
	\href{http://dx.doi.org/10.1007/978-3-319-21353-8}{Advances in Algebraic
		Quantum Field Theory.
		\newblock }.
	\newblock Mathematical Physics Studies. Springer International Publishing
	(2015)
	
	\bibitem{BrunettiFredenhagen-time:2002}
	Brunetti, R., Fredenhagen, K.: Time of occurrence observable in quantum
	mechanics.
	\newblock Phys. Rev. A \textbf{66}, 044,101 (2002)
	
	\bibitem{BrFrImRe:2014}
	Brunetti, R., Fredenhagen, K., Imani, P., Rejzner, K.: The locality axiom in
	quantum field theory and tensor products of {$C^*$}-algebras.
	\newblock \href{http://dx.doi.org/10.1142/S0129055X1450010X}{Rev. Math. Phys.
		\textbf{26}, {1450}{010}, 10 (2014)}
	
	\bibitem{BrFrVe03}
	Brunetti, R., Fredenhagen, K., Verch, R.: The generally covariant locality
	principle: A new paradigm for local quantum physics.
	\newblock \href{http://dx.doi.org/10.1007/s00220-003-0815-7}{Commun. Math.
		Phys. \textbf{237}, 31--68 (2003)}
	
	\bibitem{BuchholzSolveen:2013}
	Buchholz, D., Solveen, C.: Unruh effect and the concept of temperature.
	\newblock Classical Quantum Gravity \textbf{30}(8), {085}{011}, 9 (2013)
	
	\bibitem{BucVer_macroscopic:2015}
	Buchholz, D., Verch, R.: Macroscopic aspects of the {U}nruh effect.
	\newblock Classical Quantum Gravity \textbf{32}(24), {245}{004}, 18 (2015)
	
	\bibitem{BuchholzVerch:2016}
	Buchholz, D., Verch, R.: Unruh versus {T}olman: on the heat of acceleration.
	\newblock Gen. Relativity Gravitation \textbf{48}(3), Art. 32, 9 (2016)
	
	\bibitem{BuschLahti-StandModel}
	Busch, P., Lahti, P.: {The standard model of quantum measurement theory:
		History and applications}.
	\newblock Found. Phys. \textbf{26}, 875--893 (1996)
	
	\bibitem{Busch2009}
	Busch, P., Lahti, P.: L{\"u}ders rule.
	\newblock In: D.~Greenberger, K.~Hentschel, F.~Weinert (eds.) Compendium of
	{Q}uantum {P}hysics,, pp. 356--358. Springer Berlin, Heidelberg{,} (2009)
	
	\bibitem{Busch_etal:quantum_measurement}
	Busch, P., Lahti, P., Pellonp\"a\"a, J.P., Ylinen, K.: Quantum measurement.
	\newblock Theoretical and Mathematical Physics. Springer, [Cham] (2016).
	\newblock \urlprefix\url{https://doi.org/10.1007/978-3-319-43389-9}
	
	\bibitem{Camassa:2007}
	Camassa, P.: Relative {H}aag duality for the free field in {F}ock
	representation.
	\newblock \href{http://dx.doi.org/10.1007/s00023-007-0341-9}{Ann. Henri
		Poincar\'e \textbf{8}, 1433--1459 (2007)}
	
	\bibitem{CrispinoHiguchiMatsas:2008}
	Crispino, L.C.B., Higuchi, A., Matsas, G.E.A.: The {U}nruh effect and its
	applications.
	\newblock Rev. Modern Phys. \textbf{80}(3), 787--838 (2008)
	
	\bibitem{Davies_QTOS:1976}
	Davies, E.B.: Quantum theory of open systems.
	\newblock Academic Press [Harcourt Brace Jovanovich, Publishers], London-New
	York (1976)
	
	\bibitem{DaviesLewis:1970}
	Davies, E.B., Lewis, J.T.: An operational approach to quantum probability.
	\newblock Comm. Math. Phys. \textbf{17}, 239--260 (1970)
	
	\bibitem{deBievreMerkli:2006}
	De~Bi\`evre, S., Merkli, M.: The {U}nruh effect revisited.
	\newblock Classical Quantum Gravity \textbf{23}(22), 6525--6541 (2006)
	
	\bibitem{DeWitt:1979}
	DeWitt, B.S.: Quantum gravity: the new synthesis.
	\newblock In: S.W. Hawking, W.~Israel (eds.) General relativity: An Einstein
	Centenary Survey. Cambridge University Press, Cambridge (1979)
	
	\bibitem{DoplicherQFM}
	Doplicher, S.: {The measurement process in local quantum physics and the EPR
		paradox}.
	\newblock Commun. Math. Phys. \textbf{357}(1), 407--420 (2018)
	
	\bibitem{DuetschFredenhagen:2000}
	Duetsch, M., Fredenhagen, K.: {Algebraic quantum field theory, perturbation
		theory, and the loop expansion}.
	\newblock Commun. Math. Phys. \textbf{219}, 5--30 (2001)
	
	\bibitem{Fewster:gauge}
	Fewster, C.J.: Endomorphisms and automorphisms of locally covariant quantum
	field theories.
	\newblock \href{http://dx.doi.org/10.1142/S0129055X13500086}{Rev. Math. Phys.
		\textbf{25}, {1350}{008}, 47 (2013)}
	
	\bibitem{Fewster_artofthestate:2018}
	Fewster, C.J.: The art of the state.
	\newblock Internat. J. Modern Phys. D \textbf{27}(11), 1843,007, 26 (2018)
	
	\bibitem{FewsterJuarezAubryLouko:2016}
	Fewster, C.J., Ju\'{a}rez-Aubry, B.A., Louko, J.: Waiting for {U}nruh.
	\newblock Classical Quantum Gravity \textbf{33}(16), {165}{003}, 25 (2016)
	
	\bibitem{FewsterRejzner_AQFT:2019}
	{Fewster}, C.J., {Rejzner}, K.: {Algebraic Quantum Field Theory - an
		introduction}.
	\newblock \href{http://www.worldcat.org/search?q=isbn:978-3-030-38940-6}{In:
		F.~Finster, D.~Giulini, J.~Kleiner, J.~Tolksdorf (eds.) Progress and Visions
		in Quantum Theory in View of Gravity.
		\newblock }. Birkh\"auser, Basel (2020).
	\newblock ArXiv:1904.04051
	
	\bibitem{FewVer:20xx}
	Fewster, C.J., Verch, R.: A nonlocal generalization of green hyperbolicity.
	\newblock In preparation
	
	\bibitem{FewVer:dynloc_theory}
	Fewster, C.J., Verch, R.: Dynamical locality and covariance: What makes a
	physical theory the same in all spacetimes?
	\newblock \href{http://dx.doi.org/10.1007/s00023-012-0165-0}{{A}nnales
		H.~{P}oincar{\'e} \textbf{13}, 1613--1674 (2012)}
	
	\bibitem{FewVer:dynloc2}
	Fewster, C.J., Verch, R.: Dynamical locality of the free scalar field.
	\newblock {A}nnales H.~{P}oincar{\'e} \textbf{13}, 1675--1709 (2012)
	
	\bibitem{FewVerch_aqftincst:2015}
	Fewster, C.J., Verch, R.: Algebraic quantum field theory in curved spacetimes.
	\newblock \href{http://www.worldcat.org/search?q=isbn:978-3-319-21352-1}{In:
		R.~Brunetti, C.~Dappiaggi, K.~Fredenhagen, J.~Yngvason (eds.) Advances in
		Algebraic Quantum Field Theory.
		\newblock }, Mathematical Physics Studies, pp. 125--189. Springer International
	Publishing, Springer International Publishing (2015)
	
	\bibitem{FredenhagenHaag:1987}
	Fredenhagen, K., Haag, R.: Generally covariant quantum field theory and scaling
	limits.
	\newblock Comm. Math. Phys. \textbf{108}(1), 91--115 (1987)
	
	\bibitem{FredenhagenHaag:1990}
	Fredenhagen, K., Haag, R.: On the derivation of {H}awking radiation associated
	with the formation of a black hole.
	\newblock Comm. Math. Phys. \textbf{127}(2), 273--284 (1990)
	
	\bibitem{Giannitrapani1997}
	Giannitrapani, R.: Positive-operator-valued time observable in quantum
	mechanics.
	\newblock International Journal of Theoretical Physics \textbf{36}(7),
	1575--1584 (1997)
	
	\bibitem{GroveOttewill:1983}
	Grove, P.G., Ottewill, A.C.: Notes on ``particle detectors''.
	\newblock J. Phys. A \textbf{16}(16), 3905--3920 (1983)
	
	\bibitem{Haag}
	Haag, R.: \href{http://dx.doi.org/10.1007/978-3-642-97306-2}{Local Quantum
		Physics: Fields, Particles, Algebras.
		\newblock }.
	\newblock Springer-Verlag, Berlin (1992)
	
	\bibitem{HaagKastler1964}
	Haag, R., Kastler, D.: An algebraic approach to quantum field theory.
	\newblock J. Mathematical Phys. \textbf{5}, 848--861 (1964)
	
	\bibitem{CliftonHalvorson}
	Halvorson, H., Clifton, R.: Generic {B}ell correlation between arbitrary local
	algebras in quantum field theory.
	\newblock \href{http://dx.doi.org/10.1063/1.533253}{J. Math. Phys. \textbf{41},
		1711--1717 (2000)}
	
	\bibitem{Hawking:1975}
	Hawking, S.W.: Particle creation by black holes.
	\newblock Comm. Math. Phys. \textbf{43}(3), 199--220 (1975)
	
	\bibitem{HellwigKraus:1969}
	Hellwig, K.E., Kraus, K.: Pure operations and measurements.
	\newblock Comm. Math. Phys. \textbf{11}, 214--220 (1969)
	
	\bibitem{HellwigKraus_prd:1970}
	Hellwig, K.E., Kraus, K.: Formal description of measurements in local quantum
	field theory.
	\newblock Phys. Rev. D \textbf{1}, 566--571 (1970)
	
	\bibitem{HellwigKraus:1970}
	Hellwig, K.E., Kraus, K.: Operations and measurements. {II}.
	\newblock Comm. Math. Phys. \textbf{16}, 142--147 (1970)
	
	\bibitem{HollandsSanders:2017}
	{Hollands}, S., {Sanders}, K.: {Entanglement measures and their properties in
		quantum field theory}, \emph{Springer Briefs in Mathematical Physics},
	vol.~34 (2018).
	\newblock ArXiv:1702.04924
	
	\bibitem{Kontsevich:2003}
	Kontsevich, M.: Deformation quantization of {P}oisson manifolds.
	\newblock Lett. Math. Phys. \textbf{66}(3), 157--216 (2003)
	
	\bibitem{Kuckert:2000}
	Kuckert, B.: Localization regions of local observables.
	\newblock \href{http://dx.doi.org/10.1007/s002200000313}{Commun. Math. Phys.
		\textbf{215}, 197--216 (2000)}
	
	\bibitem{Maison:1968}
	Maison, D.: Eine {B}emerkung zu {C}lustereigenschaften.
	\newblock Comm. Math. Phys. \textbf{10}, 48--51 (1968)
	
	\bibitem{Minguzzi:2013}
	Minguzzi, E.: Convexity and quasi-uniformizability of closed preordered spaces.
	\newblock Topology Appl. \textbf{160}(8), 965--978 (2013)
	
	\bibitem{OkamuraOzawa}
	Okamura, K., Ozawa, M.: {Measurement theory in local quantum physics}.
	\newblock J. Math. Phys. \textbf{57}(1), {015}{209} (2015)
	
	\bibitem{Ozawa:1984}
	Ozawa, M.: Quantum measuring processes of continuous observables.
	\newblock J. Math. Phys. \textbf{25}(1), 79--87 (1984)
	
	\bibitem{PeresTerno:2004}
	Peres, A., Terno, D.R.: Quantum information and relativity theory.
	\newblock Rev. Modern Phys. \textbf{76}(1), 93--123 (2004)
	
	\bibitem{Rejzner_book}
	Rejzner, K.: \href{http://dx.doi.org/10.1007/978-3-319-25901-7}{Perturbative
		algebraic quantum field theory: An introduction for mathematicians.
		\newblock }.
	\newblock Mathematical Physics Studies. Springer, Cham (2016)
	
	\bibitem{SandDappHack:2012}
	Sanders, K., Dappiaggi, C., Hack, T.P.: Electromagnetism, local covariance, the
	{A}haronov-{B}ohm effect and {G}auss' law.
	\newblock Comm. Math. Phys. \textbf{328}(2), 625--667 (2014)
	
	\bibitem{Schlieder:1968}
	Schlieder, S.: Einige {B}emerkungen zur {Z}ustands\"{a}nderung von
	relativistischen quantenmechanischen {S}ystemen durch {M}essungen und zur
	{L}okalit\"{a}tsforderung.
	\newblock Comm. Math. Phys. \textbf{7}(4), 305--331 (1968)
	
	\bibitem{Smith_thesis}
	Smith, A.R.H.: Detectors, reference frames, and time.
	\newblock Springer Theses. Springer, Cham (2019).
	\newblock \doi{10.1007/978-3-030-11000-0}.
	\newblock \urlprefix\url{https://doi.org/10.1007/978-3-030-11000-0}.
	\newblock Doctoral thesis accepted independently by the University of Waterloo,
	Canada and Macquarie University, Australia
	
	\bibitem{sorkin1993impossible}
	Sorkin, R.D.: Impossible measurements on quantum fields.
	\newblock In: B.L. Hu, T.A. Jacobson (eds.) Directions in general relativity:
	Proceedings of the 1993 International Symposium, Maryland, vol.~2, pp.
	293--305. Cambridge University Press, Cambridge (1993)
	
	\bibitem{SummersWerner:1987}
	Summers, S.J., Werner, R.: Maximal violation of {B}ell's inequalities is
	generic in quantum field theory.
	\newblock Comm. Math. Phys. \textbf{110}(2), 247--259 (1987)
	
	\bibitem{SummersWerner:1995}
	Summers, S.J., Werner, R.F.: On {B}ell's inequalities and algebraic invariants.
	\newblock Lett. Math. Phys. \textbf{33}(4), 321--334 (1995)
	
	\bibitem{TolksdorfVerch:2018}
	Tolksdorf, J., Verch, R.: Quantum physics, fields and closed timelike curves:
	the {D}-{CTC} condition in quantum field theory.
	\newblock Comm. Math. Phys. \textbf{357}(1), 319--351 (2018)
	
	\bibitem{UnruhWald-what}
	Unruh, W., Wald, R.: {What happens when an accelerating observer detects a
		Rindler particle}.
	\newblock Phys. Rev. \textbf{D29}, 1047--1056 (1984)
	
	\bibitem{Unruh:1976}
	Unruh, W.G.: Notes on black-hole evaporation.
	\newblock Phys. Rev. D \textbf{14}, 870--892 (1976)
	
	\bibitem{VerchWerner:2005}
	Verch, R., Werner, R.F.: Distillability and positivity of partial transposes in
	general quantum field systems.
	\newblock Rev. Math. Phys. \textbf{17}(5), 545--576 (2005)
	
	\bibitem{Werner-TOA}
	Werner, R.: {Arrival time observables in quantum mechanics}.
	\newblock Ann. Inst. Henri Poincar\'e \textbf{47}, 429--449 (1987)
	
\end{thebibliography}
}
\end{document}